\documentclass[a4paper,11pt]{article}
\usepackage{amsmath,amsthm,mathrsfs}
\usepackage{dsfont}
\usepackage{etex}
\usepackage[utf8]{inputenc}
\usepackage{mathtools,amsmath,amssymb}
\usepackage{lmodern}
\usepackage[T1]{fontenc}
\usepackage{microtype}
\usepackage[pagebackref=false]{hyperref}
\numberwithin{equation}{section}

\usepackage[dvipsnames]{xcolor}

\usepackage{xspace}
\usepackage{bibentry}

\usepackage{graphicx}

\hypersetup{pdfborder={0 0 0},colorlinks,breaklinks=true,
  urlcolor={BurntOrange},citecolor={SeaGreen},linkcolor={BlueViolet}
}

\usepackage{subcaption}
\usepackage{tikz}
\usetikzlibrary{decorations.pathreplacing}
\usepackage[margin=2.5cm]{geometry}

\newcommand{\cris}[1]{{{\color{black}{}  #1}}}

\newcommand{\implicit}[2]
{\varphi_s(#1, #2)}

\def\be{\begin{equation}}
\def\ee{\end{equation}}
\def\bea{\begin{eqnarray}}
\def\eea{\end{eqnarray}}

\def\<{\langle}
\def\>{\rangle}
\def\~{\tilde}

\def\a{\alpha}

\newcommand{\R}{\mathbb{R}}

\newcommand{\C}{\mathbb{C}}
\newcommand{\N}{\mathbb{N}}

\newtheorem{theorem}{Theorem}[section]

\newtheorem{definition}[theorem]{Definition}
\newtheorem{lemma}[theorem]{Lemma}
\newtheorem{corollary}[theorem]{Corollary}

\newtheorem{proposition}[theorem]{Proposition}

\newtheorem*{remark}{Remark}

\newcommand{\ID}{I}

\DeclareMathOperator{\tr}{tr}

\newcommand{\oa}{{a}}

\newcommand{\oad}{{\bar{a}}}

\usetikzlibrary{decorations.markings}

    \newcommand{\parti}{
\begin{tikzpicture}
\foreach \a in {1,2,3,4,5} {
    \begin{scope}[shift={(0.5*\a,-3)}]
      \draw (0,0) rectangle (0.5,0.5);
    \end{scope}
  }
  \foreach \a in {1,2,3,4,5} {
    \begin{scope}[shift={(0.5*\a,-2.5)}]
      \draw (0,0) rectangle (0.5,0.5);
    \end{scope}
  }
  \foreach \a in {4,5} {
    \begin{scope}[shift={(0.5*\a,-2)}]
      \draw (0,0) rectangle (0.5,0.5);
    \end{scope}
  }
  \foreach \a in {4,5} {
    \begin{scope}[shift={(0.5*\a,-1.5)}]
      \draw (0,0) rectangle (0.5,0.5);
    \end{scope}
  }
  \foreach \a in {5} {
    \begin{scope}[shift={(0.5*\a,-1)}]
      \draw (0,0) rectangle (0.5,0.5);
    \end{scope}
  }
  \foreach \a in {5} {
    \begin{scope}[shift={(0.5*\a,-0.5)}]
      \draw (0,0) rectangle (0.5,0.5);
    \end{scope}
  }
  \foreach \a in {5} {
    \begin{scope}[shift={(0.5*\a,0)}]
      \draw (0,0) rectangle (0.5,0.5);
    \end{scope}
  }
  \node [below right] at (1,1) {$m_7=1$};
  \node [below right] at (0.5,-0.5) {$m_4=1$};
  \node [below right] at (-1,-1.5) {$m_2=3$};
  
  \node [below left] at (1.1,-3.5) {$x_{1}$};
  \node [below left] at (2.2,-3.55) {$\ldots$};
  \node [below left] at (3.15,-3.5) {$x_5$};
  
\end{tikzpicture}}

    \newcommand{\subparti}{
\begin{tikzpicture}
 \draw[fill=gray!30] (0.5,-2) rectangle (1,-3);
 \draw[fill=gray!30] (1.5,-2) rectangle (2,-3);
 \draw[fill=gray!30] (2.5,0.5) rectangle (3,-3);
\foreach \a in {1,2,3,4,5} {
    \begin{scope}[shift={(0.5*\a,-3)}]
      \draw(0,0) rectangle (0.5,0.5);
    \end{scope}
  }
  \foreach \a in {1,2,3,4,5} {
    \begin{scope}[shift={(0.5*\a,-2.5)}]
      \draw (0,0) rectangle (0.5,0.5);
    \end{scope}
  }
  \foreach \a in {4,5} {
    \begin{scope}[shift={(0.5*\a,-2)}]
      \draw (0,0) rectangle (0.5,0.5);
    \end{scope}
  }
  \foreach \a in {4,5} {
    \begin{scope}[shift={(0.5*\a,-1.5)}]
      \draw (0,0) rectangle (0.5,0.5);
    \end{scope}
  }
  \foreach \a in {5} {
    \begin{scope}[shift={(0.5*\a,-1)}]
      \draw (0,0) rectangle (0.5,0.5);
    \end{scope}
  }
  \foreach \a in {5} {
    \begin{scope}[shift={(0.5*\a,-0.5)}]
      \draw (0,0) rectangle (0.5,0.5);
    \end{scope}
  }
  \foreach \a in {5} {
    \begin{scope}[shift={(0.5*\a,0)}]
      \draw (0,0) rectangle (0.5,0.5);
    \end{scope}
  }
  \node [below right] at (1,1) {$\eta_7=1$};
  \node [below right] at (-1,-1.5) {$\eta_2=2$};
  
  \node [below left] at (1.1,-3.5) {$x_{1}$};
  \node [below left] at (2.2,-3.55) {$x_3$};
  \node [below left] at (3.15,-3.5) {$x_5$};
  
\end{tikzpicture}}

\begin{document}

\begingroup\parindent0pt
\centering
\begingroup\LARGE
\bf
Exact solution of an integrable \\
non-equilibrium particle system
\par\endgroup
\vspace{3.5em}
\begingroup\large
{\bf Rouven Frassek}$\,^{a,b}$, 
{\bf Cristian Giardin\`{a}}$\,^b$
\par\endgroup
\vspace{2em}
\begingroup\sffamily\footnotesize

$^a\,$Laboratoire de Physique de l’École Normale Supérieure, \\CNRS, Université PSL, Sorbonne Universités, \\24 rue Lhomond, 75005 Paris, France\\
\vspace{1em}
$^b\,$University of Modena and Reggio Emilia, FIM,\\
Via G. Campi 213/b, 41125 Modena, Italy\\
\vspace{1em} 
\par\endgroup
\vspace{5em}


\vspace{1.cm}
\begin{abstract}
\noindent
We consider the integrable family of
symmetric boundary-driven interacting particle systems 
that arise from the non-compact XXX 
Heisenberg model in one dimension 
with open boundaries.
In contrast to the well-known symmetric exclusion process,  the number
of particles at each site is unbounded.
 We show that a finite chain of $N$ sites connected at its ends to two  reservoirs can be solved exactly, 
i.e. the factorial moments of the non-equilibrium steady-state can be written in closed form  for each $N$. 
The solution relies on probabilistic arguments and techniques inspired by integrable systems. 
It is obtained in two steps:  i) the introduction of a dual absorbing process reducing the problem to a finite
number of particles; ii) the solution of the dual dynamics exploiting a symmetry obtained
from the Quantum Inverse Scattering Method. 
Long-range correlations are computed in the finite-volume system.
The exact solution allows  to prove by a direct computation that, in  the thermodynamic limit, the system approaches local equilibrium.
A by-product of the solution is the algebraic construction of a direct mapping between the 
non-equilibrium steady state
and
the equilibrium reversible measure.
%

\end{abstract}

\endgroup

\thispagestyle{empty}
\setcounter{tocdepth}{2}

\newpage

\section{Introduction}

Systems of interacting particles with stochastic dynamics
have been studied for a long time in non-equilibrium
statistical mechanics \cite{liggett2012interacting,demasi2006mathematical}. Of particular interest is the case in which the system is subject to a driving that 
produces a non-zero current in the steady state.
Restricting for simplicity to one-dimension chains with 
nearest neighbours interaction, we shall consider here the
set-up of {\em boundary-driven systems}. These are open
chains on a segment of finite length, where particles 
move in the bulk jumping with identical rates to the right or 
to the left, however the system is in contact with two 
reservoirs at its ends, injecting/removing particles at 
different rates, therefore creating a current.

\medskip
Despite their simplicity, relatively few exactly soluble cases of boundary-driven interacting particle systems
are known (e.g. the simple symmetric exclusion process \cite{Derrida_1993,Derrida2007}, the zero-range process \cite{demasi-ferrari, levine-mukamel-schutz}, the Ginzburg-Landau model \cite{demasi-olla-presutti}), especially when the invariant 
measure is not a product measure.
One of the main aim of exact solutions is to understand differences
and similarities between non-equilibrium steady states
and their equilibrium counterparts.
In one-dimensional boundary-driven systems, this amounts to comparing 
the process having reservoirs with different densities at the two ends to the process with identical reservoirs.

\medskip
The purpose of this paper is to present the solution
of a {\em family of models that can be exactly solved}: 
the boundary-driven particle systems introduced in \cite{Frassek:2019vjt}
related to the open non-compact Heisenberg XXX model in one dimension.
Our solution of the boundary-driven model relies on the combination of probabilistic arguments together with
techniques inspired by integrable systems.
Previous studies of the non-compact XXX Heisenberg spin chain include \cite{Lipatov:1993yb,Faddeev:1994zg,Derkachov:1999pz,Beisert:2003jj}. 
We further remark that the bulk-driven asymmetric version of our models has been previously studied by Bethe ansatz in \cite{Sasamoto,Povolotsky, barraquand}, in the context of KPZ universality class.

\medskip
As a by-product of our exact solution we will have the 
additional result of algebraically constructing a 
{\em  mapping between the non-equilibrium and the equilibrium processes}. 
More precisely we shall provide a (non-local) transformation between the generator of the 
non-equilibrium process and the generator of the equilibrium one with equal 
densities at the boundaries. Macroscopically, this mapping was previously observed  
by Tailleur, Kurchan and Lecomte \cite{Tailleur} in the context of the Macroscopic Fluctuation Theory
\cite{jona}, where a non-local form of the density large deviation functional emerges.

\medskip
We describe here the simplest instance belonging to the 
family of exactly solvable processes considered in this paper.
For reasons that will be clear below, we shall call 
this example the {\em open symmetric harmonic process}.
Consider a one-dimensional system of $N$ sites. 
Each site $i$, $1 \le i \le N$, is occupied by an integer 
number $m_i$ of particles. The system evolves 
(in continuous time, using independent clocks with 
exponential laws) according to the following rules:
\begin{itemize}
\item 
in the bulk, a pile of $k$ particles ($1\le k \le m_i)$ 
is moved  from site $i$ ($2\le i \le N-1$)
either to site $i+1$ or to site $i-1$, with rate $1/k$;
\item 
at the left boundary, a pile of $k$ particles is moved from site $1$ and is 
either placed at site $2$ or is removed from the system, with  rate $1/k$ ($1\le k \le m_1)$. Furthermore, a stack of $k$ particles 
($k \ge 1$) is created at site $1$ at rate $\beta_L^k/k$,
with $0<\beta_L<1$. This models the contact with a reservoir 
with an average density of particles $\rho_L = \beta_L/(1-\beta_L)$;
\item 
similarly, at the right boundary, 
a pile of $k$ particles is moved 
from site $N$ and is 
either placed at site $N-1$ or is removed from the system with rate $1/k$ ($1\le k \le m_N)$. Moreover, $k$ particles 
($k \ge 1$) are created at site $N$ at rate $\beta_R^k/k$,
with $0<\beta_R<1$. This models the interaction with a reservoir 
with an average density of particles $\rho_R = \beta_R/(1-\beta_R)$  
\end{itemize}

The name ``symmetric harmonic process'' is explained by
considering the amount of time spent by the system
in the configuration $(m_1,m_2,\ldots,m_N)$ before 
jumping to a new configuration, which
results in an  exponentially 
distributed random variable with parameter 
 $ \sum_{i=1}^{N} 2 h(m_i)  - \log(1-\beta_L)  - \log(1-\beta_R)$,
where
\begin{equation}
\label{eq:harmonic}
h(n) = \sum_{k=1}^n \frac1k
\end{equation}
is the $n$-th harmonic number. 
In the more general version of the exactly solvable
process, the rate $1/k$ (with $1\le k \le n$) at which a pile of $k$ particles is moved
from a site containing $n$ particles, 
is replaced by a suitable function $\implicit{k}{n}$
(see Section~\ref{sec:model}) labelled by the
parameter $s>0$. {The process is symmetric since the rates of particles moving to the left coincide with the rates of particles moving to the right. }

\medskip
These models are interesting for a number of reasons.
Unlike the equilibrium case ($\beta_L=\beta_R$), 
the density profile in the non-equilibrium 
steady state ($\beta_L \neq \beta_R$) is not uniform,
and is related to a non-zero current by the
Fourier's law. Being the interaction among particles
of zero-range type (the transition rates are just  
functions of the number of particles at the departure site),
one might think that their non-equilibrium 
stationary state is factorized, as
it happens in the ordinary zero-range
process \cite{demasi-ferrari, levine-mukamel-schutz}. Instead, as we shall see,
they provide examples of non-equilibrium systems with long-range spatial and temporal correlations.
The difference with respect to the ordinary 
zero range-process arises from the movement 
of piles of particles, that prevents factorization
of the steady state. Rather,
 truncated correlation between 
particle numbers in the stationary state will be proved to be non-zero and actually positive.

\medskip
In this paper we rigorously study the invariant measure
of the open symmetric harmonic process, and its
generalized version as well. 
\cris{More precisely,
the factorial moments of the stationary state of a finite chain of length $N$ are obtained. The stationary state can then be reconstructed 
for all values of the boundaries parameters 
$\rho_L,\rho_R$ by using a standard inversion formula.}
By taking the thermodynamic
limit $N\to\infty$, approach to the Gibbs distribution is proven
and a linear density profile $\rho(u)= \rho_L + u(\rho_R-\rho_L)$ 
is obtained. We prove that around each macroscopic point $0\le u \le 1$,
which corresponds to the micro-point $\lfloor u N\rfloor$, there
is a local equilibrium distribution
with density $\rho(u)$. 
Microscopically, we show that correlations are weak but long-ranged
and the stationary measure is far from a product measure.
Hence we expect non-trivial properties of fluctations,
with a non-local structure of the density large deviation
functional, as predicted by the Macroscopic Fluctuation Theory \cite{jona}.

\medskip
Our proof starts from a duality relation with an absorbing dual process \cite{kipnis1982heat,schutz1997duality}.
After this we compute the absorption probabilities of dual
particles by using a non-trivial symmetry that is identified
by means of the Quantum Inverse Scattering Method \cite{Faddeev:1996iy,Korepin:1993kvr}.
As a result, we find closed expressions for the factorial moments
of the non-equilibrium steady state.
The mere existence of an absorbing dual implies
a (non-local) mapping between equilibrium and non-equilibrium stationary state.
The open symmetric harmonic process studied here 
is special in that the mapping can be fully elucidated at microscopic level from its algebraic representation. To our knowledge,
the only known other model where this has been achieved is the
open symmetric exclusion process \cite{Frassek:2020omo,Frassek:2019imp}, which is solvable by the Matrix Product Ansatz \cite{Derrida_1993}.

\paragraph{Organization of the article.}	
The paper is organized as follows.
In Section~\ref{model} we introduce
the model and state our main results.
In Section~\ref{sec:algebraic} we provide the algebraic 
description of the Markov  generator in terms of the
non-compact $\mathfrak{sl}(2)$ Lie algebra generators.
In Section~\ref{sec:first step} we complement this 
with the algebraic description of the Markov generator of the dual
process and establish the duality relation.
In Section~\ref{sec:transform} we identify
the sequence of similarity transformations
that make the boundaries of the spin chain diagonal.
This is obtained by using a non-trivial symmetry
of the open chain.
Finally, in Section~\ref{sec:proof-main-th} we 
prove our main results.

\section{Model and results}
\label{model}

We start by recalling the definition of the open interacting particle systems introduced in \cite{Frassek:2019vjt}.
For any  integer $N \ge 1$, we shall denote by ${\mathscr C}_N$  the configuration space made of the $N$-dimensional vectors $m = (m_1,m_2,\ldots,m_N)$
with non-negative integer components. For $i \in\{1,\ldots,N\}$ we shall write $\delta_i$ for the vector with all components zero  except in the $i^{th}$
place, i.e.
\be
({\delta}_i)_j=
 \left\{\begin{array}{rl}
 1 & \text{if } j=i,\\
 0 & \text{otherwise}.
 \end{array} \right.
\ee
In the following we denote the set of non-negative integers by $\N_0 = \N \cup \{0\}$.

\subsection{Process definition}
\label{sec:model}

\begin{definition}[The process]
\label{def:process}
For real numbers  $\beta_L, \beta_R \in (0,1)$ and $s>0$, 
we consider the continuous-time  Markov chain 
having configuration space ${\mathscr C}_N$ and being defined by the infinitesimal generator ${\mathscr L}$
whose action on local functions $f : {\mathscr C}_N\to\R$ is given by
\be
\label{eq:gen}
{\mathscr L}f := {\mathscr L}_1 f + \sum_{i=1}^{N-1}{\mathscr L}_{i,i+1}f + {\mathscr L}_N f
\ee
where
\begin{eqnarray}
\label{eq:genbulk}
({\mathscr L}
_{i,i+1} f)(m) 
& := & 
\sum_{k=1}^{m_i}\varphi_s(k,m_i) \Big[f(m-k \delta_i + k \delta_{i+1}) - f(m)\Big] \\
& + & 
\sum_{k=1}^{m_{i+1}} \varphi_s(k,m_{i+1}) \Big[f(m+k\delta_i - k \delta_{i+1}) - f(m)\Big] \nonumber
\end{eqnarray}
and, for $i\in\{1,N\}$,
\begin{eqnarray}
\label{eq:genbdry}
({\mathscr L}
_{i} f)(m) 
& := &  
\sum_{k=1}^{m_i} \varphi_s(k,m_i) \Big[f(m-k\delta_i) - f(m)\Big]
 +   
\sum_{k=1}^{\infty} \frac{\beta_i^k}{k} \Big[f(m+k\delta_i) - f(m)\Big]
\end{eqnarray}
with $\beta_1= \beta_L$ and $\beta_N = \beta_R$.
The function $\varphi_s : \N \times \N \to \R$ is given by
\be
\label{eq:varphi}
\varphi_s(k,n) :=  \frac{1}{k}\frac{\Gamma (n+1) \Gamma (n-k+2 s)}{ \Gamma (n-k+1) \Gamma (n+2 s)} \mathds{1}_{\{1\le k\le n\}}
\ee
where, by abuse of notation, $\mathds{1}_{\{1\le k\le n\}}$ is used in place of the indicator function of the set $\{1,2,\ldots,n\}$.
We shall denote the process by $\{m(t) = (m_1(t), \ldots, m_N(t)) : t\ge 0\}$; the component $m_i(t)$
gives the number of particles at time $t\ge 0$ at site $i\in\{1,\ldots,N\}$.  {For $N=1$ the bulk term vanishes and only the boundaries remain in the generator \eqref{eq:gen}.}
\end{definition}

\begin{remark}[Existence of the process]
{\rm
For all integers $N \ge 1$ and initial configurations $m\in {\mathscr C}_N$ the stochastic process just defined is non-explosive,
i.e. in any bounded interval of time, almost surely, the system undergoes finitely many 
transitions.  This follows from irreducibility and (positive) recurrence of the embedded 
discrete time chain.  
The total number of particles in the system $|m(t)|=\sum_{i=1}^N m_i(t)$
remains finite at all times with probability one. Indeed, in the course of time, this number is changed only by the boundary generator in \eqref{eq:genbdry},
as the dynamics in the bulk is conservative.
For any $\beta_i \in (0,1)$,  the boundary generator in \eqref{eq:genbdry}
describe a continuous-time inhomogeneous 
random-walk on $\N_0$ that is reflected in $0$ and, from the position $m_i$,  jumps 
to $m_i+k$ with $k\in \N$ at rate ${\beta_i^k}/{k}$
and jumps to $m_i-k$ with $1\le k \le m_i$ at rate $\varphi_s(k,m_i)$.
Such walker does not escape to infinity since it has a negative drift.
}
\end{remark}

\begin{remark}[Shifted harmonic numbers]
{\em 
One can  check that 
\be
\label{eq:identity}
h_s(n) := \sum_{k=1}^n \implicit{k}{n} = \psi(2s+n)-\psi(2s)
\ee
where $\psi$ is the digamma function,
i.e. the logarithmic derivative of the Gamma function, $\psi(x) = \frac{d}{dx}\ln \Gamma(x)$. 
From the
recurrence property of the digamma function,
$
\psi(x+1) = \psi(x) + {1}/{x}
$,
one immediately finds that the holding time
of the process in a configuration $m=(m_1,m_2,\ldots,m_N)$ 
is an exponentially distributed random variable with parameter 
 $ \sum_{i=1}^{N} 2 h_s(m_i)  - \log(1-\beta_L)  - \log(1-\beta_R)$,
where
\begin{equation}\label{eq:harmonicgen}
 h_s(n) =\sum_{k=1}^n\frac{1}{k+2s-1}
\end{equation} 
are the ``shifted'' harmonic numbers.
For $s=1/2$ one recovers the standard harmonic numbers in \eqref{eq:harmonic}.
}
\end{remark}
Since the Markov chain in Definition \ref{def:process} is irreducible and recurrent
then it has a unique invariant probability measure.
\begin{definition}[Stationary state] 
\label{def:stat-dist}For two configurations $m,m'\in\mathscr{C}_N$, let $L(m,m')$ denote the transition rate from $m$ to $m'$
of the process $\{m(t) \,:\; t\ge 0\}$ with generator ${\mathscr L}$ in \eqref{eq:gen}. 
The stationary state  (or invariant distribution)  $\mu$ of the process is
the $|{\mathscr C}_N|$-dimensional vector
satisfying
\be
\mu(m) = \sum_{m'\in {\mathscr C}_N}  \mu(m')L(m',m)   \qquad\qquad\qquad \forall \,m \in {\mathscr C}_N
\ee
and $
\sum_{m\in {\mathscr C}_N}  \mu(m) = 1
$. We shall denote by $\mathbb{E}[\cdot]$ the expectation with respect to the stationary state.
\end{definition}
Note that the steady state distribution $\mu$ is a function
of the system size $N$ and of the reservoir parameters
$\beta_L,\beta_R$. To alleviate the notation,
we shall not write these dependencies explicitly.

\medskip

In the following, we shall call ${\mathscr L}^{\mathrm{eq}}$ the 
generator ${\mathscr L}$ in \eqref{eq:gen} with $\beta_L=\beta_R=\beta$, 
describing the system at equilibrium. 
If $\beta_L=\beta_R=\beta$ then the particle process $\{m(t)\, , \, t\ge 0\}$ 
has an invariant measure given by a product of Negative Binomial distributions 
with parameters $0<\beta<1$ and $2s >0 $. Namely, for a chain of length $N$, the law
with mass function
\be
\label{rev-mes}
\mu^{\mathrm{eq}}({m}) :=  \prod_{i=1}^N \left[\frac{\beta^{m_i}}{m_i!} \frac{\Gamma(m_i + 2s)}{\Gamma(2s)} (1-\beta)^{2s}\right]
\ee
is reversible and thus stationary.  

{Our aim in this paper is to study the stationary measure in the non-equilibrium setting $\beta_L\neq \beta_R$}.
We remark that in the non-equilibrium case reversibility is lost
and furthermore one can check that a product ansatz for the stationary measure does not work. 
Indeed the non-equilibrium steady state will be shown to have long-range correlations.

\subsection{Basic properties}
\label{basic}

To obtain the non-equilibrium steady state we  introduce an auxiliary process.
To this aim we enlarge the configuration space by adding two absorbing extra sites $\{0,N+1\}$.
\begin{definition}[The dual process]
\label{def:processdual}
For a real number  $s>0$, 
we consider the continuous-time Markov chain $\{\xi(t) = (\xi_0(t), \xi_1(t),\ldots, \xi_{N+1}(t)) : t\ge 0\}$
having configuration space ${\mathscr C}_{N+2}$ and being defined by the infinitesimal generator ${\mathscr L}^{\mathrm{dual}}$
whose action on local functions $f : {\mathscr C}_{N+2}\to\R$ is given by
\be
\label{eq:gendual}
{\mathscr L}^{\mathrm{dual}}f :=  \sum_{i=0}^{N}{\mathscr L}^{\mathrm{dual}}_{i,i+1} f
\ee
where
\be
{\mathscr L}^{\mathrm{dual}}_{i,i+1} f(\xi) := {\mathscr L}_{i,i+1}f(\xi) \qquad\qquad \forall \, i \in \{1,2,\ldots N-1\} 
\ee
with ${\mathscr L}_{i,i+1}$ being the bond generator defined in \eqref{eq:genbulk}, and
\begin{eqnarray}
\label{eq:gendualbdry1}
{\mathscr L}^{\mathrm{dual}}
_{0,1} f(\xi) 
& := &  
\sum_{k=1}^{\xi_1} \varphi_s(k,\xi_1) \Big[f(\xi-k\delta_1 + k \delta_0) - f(\xi)\Big]\,,
\end{eqnarray}
\begin{eqnarray}
\label{eq:gendualbdry2}
{\mathscr L}^{\mathrm{dual}}
_{N,N+1} f(\xi) 
& := &  
\sum_{k=1}^{\xi_N} \varphi_s(k,\xi_N) \Big[f(\xi-k\delta_N + k \delta_{N+1}) - f(\xi)\Big]\,.
\end{eqnarray}
\end{definition}
Thus, in the dual process, particles move as in the original system while they are in the bulk; when
they reach the boundaries they can be absorbed in either one of the extra sites $\{0,N+1\}$, where they remain forever. 

We state the relation between the process and the dual process in the next proposition.
We  introduce the notation  $\mathbb{E}_{m}$ for the expectation 
with respect to the Markov process $\{m(t) : t \ge 0\}$ 
generated by ${\mathscr L}$ and starting at $m(0) = m$, i.e.
for a measurable function $f:{\mathscr C}_{N} \to \R$
\be
\mathbb{E}_{m}[f(m(t)] := \sum_{m'\in {\mathscr C}_{N}} f(m') p_t(m,m')
\ee
with $p_t(\cdot,\cdot)$ the transition probability function of the process. Similarly 
we denote  $\mathbb{E}^{\mathrm{dual}}_{\xi}$ the  expectation with respect to the 
dual Markov process $\{\xi(t): t \ge 0\}$ generated by ${\mathscr L}^{\mathrm{dual}}$ and 
starting at $\xi(0) = \xi$, i.e. for a measurable function $f:{\mathscr C}_{N+2} \to \R$
\be
\mathbb{E}^{\mathrm{dual}}_{\xi}[f(\xi(t)] := \sum_{\xi'\in {\mathscr C}_{N+2}} f(\xi') p^{\mathrm{dual}}_t(\xi,\xi')
\ee
with $p^{\mathrm{dual}}_t(\cdot,\cdot)$ the transition function of the dual process.

\begin{proposition}[Duality]
\label{prop:duality}
Denote the reservoir densities by
\be\label{eq:rho}
\rho_L = \frac{\beta_L}{1-\beta_L} 
\qquad\qquad
\rho_R = \frac{\beta_R}{1-\beta_R}.
\ee
Define the duality function $D: {\mathscr C}_{N}\times {\mathscr C}_{N+2} \to \R$
\be
\label{eq:duality-fct}
D(m,\xi) := \rho_L^{\xi_0}  \left(\prod_{i=1}^N \frac{m_i!}{(m_i-\xi_i)!} \frac{\Gamma(2s)}{\Gamma(2s+\xi_i)} \right) \rho_R^{\xi_{N+1}}.
\ee
Then, for any time $t>0$ and for every configurations $m\in{\mathscr C}_{N}$ and $\xi\in{\mathscr C}_{N+2}$, we have the  equality 
\be
\label{eq:duality-relation}
\mathbb{E}_{m} [D(m(t),\xi)] = \mathbb{E}^{\mathrm{dual}}_{\xi} [D(m,\xi(t))]. 
\ee
\end{proposition}

\medskip

Proposition \ref{prop:duality} was proved in \cite{Frassek:2019vjt}  for $s=1/2$ by an explicit computation
and stated without proof for generic $s>0$. We shall provide an algebraic independent proof in Section~\ref{sec:first step} that applies to all $s>0$.

We will characterize the non-equilibrium steady state of our system by computing the (scaled) 
factorial moments,
whose definition we hereafter provide.
\begin{definition}[Scaled factorial moments]
\label{def:fac}
For a multi-index $\xi=(\xi_1,\ldots,\xi_N) \in \N_0^N$, 
the  scaled factorial moments 
of order $|\xi| = \sum_{i=1}^N \xi_i$ of the invariant distribution
$\mu$ in Definition \ref{def:stat-dist}
are given by
\be
\label{eq:fac-moments}
G(\xi) := \sum_{{m}\in {\mathscr C}_N}  \mu({m}) 
\Big[\prod_{i=1}^N 
\frac{m_i}{2s}
\cdot
\frac{(m_i-1)}{(2s+1)} 
\cdots 
\frac{(m_i-\xi_i+1)}{(2s+\xi_i-1)}
\Big]\;.
\ee
\end{definition}
Compared to the standard textbook definition of factorial moments,
here the rising factorials $2s(2s+1)\ldots (2s+\xi_i-1)$ have been
added in the denominator. This explains the name ``scaled'' factorial moments.
The reason for this choice is that,
as a consequence of Proposition \ref{prop:duality}, 
the computation of the scaled factorial moments 
admits a direct probabilistic description in terms 
of the dual process.

\begin{proposition}[Scaled factorial moments via absorption probabilities]
\label{prop:facmom-abs}
For a multi-index $\xi = (\xi_1,\ldots,\xi_N) \in \N_0^N$, let $\check{\xi} \in \mathscr{C}_{N+2}$ be the con\-fi\-gu\-ration of the dual process defined by $\check{\xi}= (0,\xi_1,\ldots,\xi_N,0)$.
 Denote by $|\xi| = \sum_{i=1}^N \xi_i$ the total number of  particles in $\check{\xi}$.
Then the scaled factorial moments  are given by
\be
\label{eq:expansion}
G(\xi) = \sum_{k=0}^{|\xi|} \rho_L^k \, \rho_R^{|\xi|-k} \, p_{\check{\xi}}(k)\,,
\ee
where the reservoir densities $\rho_L,\rho_R$ are defined in \eqref{eq:rho} and $p_{\check{\xi}}(k)$ is the probability that, if the dual process is started 
from the configuration $\check{\xi}$, then $k$ of the $|\xi|$ particles are eventually  absorbed at $0$ and the remaining $|\xi|-k$ 
are absorbed at $N+1$, i.e.
\be
\label{eq:abs-prob}
p_{\check{\xi}}(k) = \mathbb{P}\big[\xi(\infty) = k \delta_0 + (|\xi|-k) \delta_{N+1} \,|\, \xi(0) = \check\xi\big]\,, \qquad\qquad 0\le k \le |\xi|\,.
\ee
\end{proposition}
The proof of the above proposition will also be given in Section~\ref{sec:first step}.

\subsection{Main results} 
\label{sec:factorial}

We now present our results which describe the non-equilibrium steady state of the open
symmetric harmonic process. The main finding is the following closed-form expression for the scaled factorial moments.
\begin{theorem}[Scaled factorial moments]
\label{theo:factorial}
For a multi-index $\xi=(\xi_1,\ldots,\xi_N)\in\N_0^N$, the scaled factorial moments of the non-equilibrium steady state 
are given by
\begin{equation}
\label{ed2}
G(\xi) =\sum_{n=0}^{|{ \xi}|}\rho_R^{|{ \xi}|-n}(\rho_{L}-\rho_{R})^n g_{\xi}(n)
\end{equation} 
with
\begin{equation}\label{eq:gfunc}
 g_{\xi}(n)=
\sum_{{\underset{\eta_1+\ldots+\eta_{\scalebox{0.5} N}=n}{(\eta_1,\ldots,\eta_{\scalebox{0.5} N}) \in \N_0^{\scalebox{0.5} N}}}}
\prod_{i=1}^N \binom{\xi_i}{\eta_i} \prod_{j=1}^{\eta_i}\frac{2s(N+1-i)-j+\sum_{k=i}^{N}\eta_k}{2s(N+1)-j+\sum_{k=i}^{N}\eta_k}\,.
\end{equation}
In particular,
\begin{equation}
\label{eq: express-one}
 g_{\xi}(n)=
\sum_{|\eta|=n}\prod_{i=1}^N \binom{\xi_i}{\eta_i} \frac{N+1-i}{N+1-i +\sum_{k=i}^{N}\eta_k}
\end{equation}
 for the case $s=1/2$.

\end{theorem}

\begin{remark}[Half-integer spin values]
{\em
As explained in the proof of Theorem \ref{theo:factorial}, a simplification similar to the one in \eqref{eq: express-one} occurs for all $s\in \N/2$, see \eqref{eq: express-onexxx}.
}
\end{remark}

\begin{remark}[Absorption probabilities]
{\em
By expanding $(\rho_{L}-\rho_{R})^n$ in \eqref{ed2}, the scaled  factorial moments can  be written as a polynomial in the densities of the two reservoirs
as in \eqref{eq:expansion}.
The coefficients $p_{\check\xi}(\cdot)$,
that are obviously related to the coefficients
$g_{ \xi}(\cdot)$ by 
\begin{equation}
p_{\check\xi}(k)=
\sum_{n=k}^{|\xi|}(-1)^{n-k}\binom{n}{k}
 g_{ \xi}(n)
 \,,
\end{equation} 
are the absorption probabilities  \eqref{eq:abs-prob}. 
}
\end{remark}

The multi-index that labels the factorial moments 
can equivalently be prescribed by assigning the {\em ordered} positions of the dual particles. 
Namely, a configuration ${\xi}= (\xi_1,\ldots,\xi_N)$ with
$|{ \xi}|=\sum_{i=1}^N\xi_i$ particles can alternatively
be described by their ordered positions 
$1\le x_1 \le x_2 \le \ldots \le x_{|{ \xi}|}\le N$, i.e. 
$$
\xi = \sum_{i=1}^{|\xi|} \delta_{x_i}.
$$
We then have the following alternative expression for the scaled factorial moments.
\begin{corollary}[Scaled factorial moments, coordinate form]
\label{cor:fac-mom}
For a multi-index  $\xi= (\xi_1,\ldots,\xi_N)\in\N_0^N$ that is in bijection
with the ordered set of positions $x=(x_1,\ldots,x_{|\xi|})$ satisfying the condition  $1\le x_1 \le \ldots \le x_{|{ \xi}|}\le N$, we have
\begin{equation}
\label{ed}
G(x)= \sum_{n=0}^{|\xi|}\rho_R^{|  \xi|-n}(\rho_L-\rho_R)^n\
g_{x}(n)
\end{equation} 
with
\begin{equation}
\label{eq: express-two}
 g_{x}(n)=
\sum_{1\leq i_1< \ldots< i_n\leq |  \xi|} \; \prod_{\alpha=1}^n\frac{n-\alpha+2s(N+1-x_{i_\alpha})}{n-\alpha+2s(N+1)}\,.
\end{equation} 
\end{corollary}

We continue by observing that, knowing the factorial moments, one can reconstruct the stationary measure
of the open symmetric harmonic process using the inversion formula \cite{feller2008introduction}
\begin{equation}\label{eq:measure}
\mu(m) = \sum_{\xi \ge m} G(\xi)  \Big[\prod_{i=1}^N \frac{(-1)^{\xi_i-m_i}}{\xi_i!} {\xi_i \choose m_i} \frac{\Gamma(2s+\xi_i)}{\Gamma(2s)} \Big]\,,
\end{equation}
where the restriction on the summation has to interpreted componentwise, i.e. 
$\xi_i\ge m_i$ for all $i\in\{1,\ldots,N\}$.
Inserting the factorial moments \eqref{ed2}-\eqref{eq:gfunc} into \eqref{eq:measure} we obtain the following expression for the  non-equilibrium steady state in terms of the parameters $\beta_{L}$ and $\beta_R$ .

 \begin{corollary}[Stationary state]
 \label{theo:steady}
 For $m\in \N_0^N$, the weights of the invariant distribution in Definition 
\ref{def:stat-dist} are given by
 \be
 \label{eq:stead1}
 \mu({m}) = 
 \sum_{\xi \in \N_0^N}
 \sum_{\eta \in \N_0^N}
 \rho_R^{|{ \xi}|-|\eta|} (\rho_L-\rho_R)^{|\eta|}
 \varphi_{{ m}}(\xi,\eta)
 \ee
 where
 \be
 \label{eq:stead2}
 \varphi_{{ m}}({ \xi},\eta)=
 \prod_{i=1}^N {\xi_i \choose \eta_i} {\xi_i \choose m_i} 
 \frac{(-1)^{\xi_i - m_i} }{\xi_i!}\frac{\Gamma(2s+\xi_i)}{\Gamma(2s)} 
  \prod_{j=1}^{\eta_i}\frac{2s(N+1-i)-j+\sum_{k=i}^{N}\eta_k}{2s(N+1)-j+\sum_{k=i}^{N}\eta_k}
 \ee
 \end{corollary}

We notice that, in  formula \eqref{eq:stead1}, the sum over $\eta$
is actually a finite sum due the Newton binomial coefficient in $\varphi_{{ m}}({ \xi},\eta)$. 
\begin{remark}[Case $s=1/2$ and $N=1,2$]
{\em { Similar to the scaled factorial moments,  the stationary measure \eqref{eq:stead1}-\eqref{eq:stead2}  simplifies for spin $1/2$, cf.~\eqref{eq: express-one}. In this case we can evaluate the sums over the $\eta$ variables by performing the change of variables $\eta_i\to\eta_i-\eta_{i+1}$ where $i=1,\ldots,N-1$. 
For $N=1,2$ this leads to the following results.}
In the case of one site $N=1${, where the bulk contribution is absent,} and spin $s=1/2$ the stationary measure reads
\begin{equation}
 \mu(m_1)=\frac{(\beta_L-1) (\beta_R-1)}{\beta_L-\beta_R} \left(\sum_{k=m_1+1}^\infty\frac{\beta_L^k}{k}-\sum_{k=m_1+1}^\infty\frac{\beta_R^k}{k}\right).
\end{equation} 
For two sites and $s=1/2$ we obtain the stationary measure
\begin{equation}
 \mu(m_1,m_2)=2\frac{(\beta_L-1)^2(\beta_R-1)^2}{(\beta_L-\beta_R)^2}\left(  \phi_{\beta_L}(m_1,m_2)- \kappa(m_1,m_2)   + \phi_{\beta_R}(m_2,m_1)    \right)
\end{equation} 
where
\begin{equation}
\begin{split}
 \phi_\beta(m_1,m_2)=\frac{1}{2}\gamma_{\beta}^2(1+m_1)&-\sum_{k=m_1+1}^{m_2}\frac{1}{k}    \, \gamma_{\beta}(m_1+k+1)+\sum_{k=m_2+1}^{m_1}\frac{1}{k}  \,   \gamma_{\beta}(m_1+k+1)
 \end{split}
\end{equation} 
and
\begin{equation}
 \kappa(m_1,m_2)=\gamma_{\beta_L}(1+m_1)\gamma_{\beta_R}(1+m_2)   \,.
\end{equation} 
Here we introduced the sum
\begin{equation}
 \gamma_\beta(n)=\sum_{k=n}^\infty \frac{\beta^k}{k}=-\log(1-\beta)-\sum_{k=1}^{n-1}\frac{\beta^k}{k}\,,
\end{equation} 
which is related to the incomplete Beta function via $B_{\beta}(n,0)= \gamma_\beta(n)$.
}
\end{remark}

The exact solution at finite volume $N$ allows to establish, by a direct 
computation, that in the thermodynamic limit $N\to\infty$
the non-equilibrium stationary measure approaches locally a Gibbs 
distribution and transport of mass across the system satisfies Fick's law.
Let $O$ be the algebra of cylindrical bounded functions on $\N_0^{\N}$
and denote by $\tau_i$ the translation by $i$, i.e. for all function $f\in O$
define $(\tau_i f)(j) = f(i+j)$.

\begin{corollary}[Local equilibrium \& Fick's law] 
Let $\rho_L,\rho_R$ be defined by 
\eqref{eq:rho} and let $\mu$ be the unique invariant measure for the 
open symmetric harmonic process.
Then the following hold: 
\begin{itemize}
\item[(i)]
for $u\in(0,1)$
\label{theo:loc}
\be
\lim_{N\to\infty} \mu(\tau_{[u N]} f) = \nu_{\rho(u)}(f) \qquad\qquad \forall \, f \in O
\ee
where $[x]$ denotes the integer part of $x\in\R$,  
\be
\label{eq:linear}
\rho(u) = \rho_L + (\rho_R-\rho_L) u \,,
\ee
and
$\nu_{\rho}$ is the product measure on $\N_0^{\N}$ with marginals
given by Negative Binomial distributions with shape parameter $2s$
and mean $\rho$;
\item[(ii)]
in a system of size $N$, define the stationary current between
two neighbor sites $i,i+1$ by
\be
\label{eq:currbond}
J_{i,i+1} =  \sum_{m\in\mathscr{C}_N} \mu(m) [m_i - m_{i+1}]
\ee
and the total stationary current in the thermodynamic limit as
\be
J = \lim_{N\to\infty} \sum_{i=1}^{N-1} J_{i,i+1}.
\ee
Then Fick's law holds, namely
\be
\label{eq:fick}
J = - K_s\frac{d\rho(u)}{du}  \qquad\qquad u\in(0,1)
\ee
where the diffusivity $K_s = 2s$ and the density profile $\rho(u)$ is defined by Eq. \eqref{eq:linear}.
\end{itemize}
\end{corollary}

\subsection{Discussion and open problems}
\label{sec-disc-op}
In this section, we discuss our results, considering in particular relations to the literature, possible extensions and open problems.

\paragraph{Local equilibrium.} The property of local equilibrium (Corollary \ref{theo:loc})
is  important for the construction of non-equilibrium thermodynamics
\cite{oono1998steady,Sasa,jona}. 
Indeed, as a consequence of local equilibrium, it is possible to locally define 
thermodynamic variables (such as the density) which vary 
smoothly on the macroscopic scale. In boundary driven systems with an absorbing dual,
proving local equilibrium amounts to prove that the absorption probabilities of the dual particles factorize on the macroscopic scale. One way to achieve this is via  the construction of a coupling between the dual particles and independent particles, which is often a non-trivial task (see, e.g., \cite{demasi2006mathematical} for
the symmetric exclusion process and  \cite{kipnis1982heat} for the KMP model).
Our system is special in that we can solve the asymptotic dual dynamics (due to integrability of the model)  and thus
we can directly prove local equilibrium as a consequence of the exact solution.

\paragraph{Correlation functions and cumulants.}
The correlation functions 
in the non-equilibrium steady state can be read off
from Corollary \ref{cor:fac-mom}. 
In particular, if ${ \xi}=(\xi_1,\ldots,\xi_N)\in\mathbb{N}_0^N$,
then $G(\xi)$ provides information
on the $|{ \xi}|$-point correlation functions in the
steady state. \cris{The correlation functions turn out to be 
very long-ranged and contribute, despite their vanishing pointwise as 
$N\to \infty$, to the fluctuations about the typical density profile.
}
We show this by considering
a few simple cases. \cris{In this paragraph and in the next one}, $M=(M_1,\ldots,M_N)$
denotes a random vector whose probability distribution is the
non-equilibrium steady state.

Using  $\xi ={\delta}_{x_1}$ (one dual particle at position $1\le x_1\le N$), we obtain the average profile \cris{$\mathbb{E}[M_{x_1}]= 2s G(x_1)$}
that linearly interpolates between the two densities $\rho_L$ and $\rho_R$ at the boundaries:
\be
\label{profile}
\mathbb{E}[M_{x_1}] =  2s\left(\rho_L + \frac{\rho_R-\rho_L}{N+1} x_1\right).
\ee

\noindent
Considering  $\xi ={\delta}_{x_1}+  {\delta}_{x_2}$ (two dual particles at positions $x_1$ and $x_2$)
one gets the joint cumulants of second order
\be
\kappa_2(M_{x_1},M_{x_2}) = \mathbb{E}[M_{x_1} M_{x_2}] - \mathbb{E}[M_{x_1}]  \mathbb{E}[M_{x_2}].
\ee 
In particular, for $1\le x_1 < x_2 \le N$, we obtain the covariance \cris{$\mathbb{C}ov[M_{x_1},M_{x_2}]= (2s)^2 [G(x_1,x_2) - G(x_1)G(x_2)]$ which reads}
\begin{eqnarray}
\label{cov}
\mathbb{C}ov(M_{x_1}, M_{x_2})
& = & 
(2s)^2  \frac{ x_1(N+1-x_2)}{(N+1)^2(1+2s(N+1))} (\rho_R-\rho_L)^2\,,
\end{eqnarray}
and for $1\le x_1 = x_2 \le N$ we obtain the variance
\cris{$\mathbb{V}ar[M_{x_1}]= (2s)^2 [G(x_1,x_1) - G(x_1)^2] + 2s [G(x_1,x_1) + G(x_1)]$ which reads}
\begin{eqnarray}
\label{var}
Var(M_{x_1})
& = & 
\frac{2s}{2s(1+N)^3 + (1+N)^2}
\Big((N+1)^2(1+2s(N+1))\rho_L(1+\rho_L)\nonumber\\
&&-(N+1)(\rho_L-\rho_R)(1+\rho_L+\rho_R+2s(1+N+\rho_L+2N\rho_L+\rho_R))x_1 \nonumber\\
&&+ 2s N (\rho_L-\rho_R)^2 x_1^2\Big)\,.
\end{eqnarray}

\noindent
Using  $\xi ={\delta}_{x_1}+  {\delta}_{x_2} +  {\delta}_{x_3}$  (three dual particles), 
we find the third order cumulant 
\begin{eqnarray}
\kappa_3(M_{x_1},M_{x_2},M_{x_3}) & = & 
\mathbb{E}[M_{x_1} M_{x_2} M_{x_3}] +2 \mathbb{E}[M_{x_1}]  \mathbb{E}[M_{x_2}] \mathbb{E}[M_{x_3}]  \\
&-& \mathbb{E}[M_{x_1} M_{x_2}] \mathbb{E}[M_{x_3}] 
- \mathbb{E}[M_{x_1} M_{x_3}] \mathbb{E}[M_{x_2}] 
- \mathbb{E}[M_{x_2} M_{x_3}] \mathbb{E}[M_{x_1}]. \nonumber
\end{eqnarray} 
For $1\le x_1 < x_2 < x_3\le N$ the result is
\be
\label{k3}
\kappa_3(M_{x_1},M_{x_2},M_{x_3}) =  (2s)^3  \frac{x_1(N+1-2 x_2)(N+1-x_3)}{(N+1)^3(1+s+sN)(1+2s(N+1))} (\rho_R-\rho_L)^3.
\ee
Similarly, considering $\xi= \sum_{i=1}^n \delta_{x_i}$ one can compute
the $n^{\text{th}}$   cumulant 
for which it is expected \cite{DLS} that
\be
\kappa_n(M_{x_1},\ldots,M_{x_n}) = f^{(N)}_{n}(x_1,\ldots,x_n) (\rho_R-\rho_L)^n\,,
\ee
with $f^{(N)}_{n}$ a sequence of functions such that
\be
\lim_{N\to\infty} N^{n-1} f^{(N)}_{n}(Ny_1,\ldots,Ny_n) = f_{n}(y_1,\ldots,y_n)\,,
\ee
with $0<y_1<y_2<\ldots<y_n <1$. 
\cris{The scaling of the $n^{\text{th}}$ cumulant as the $n^{\text{th}}$ power of the density difference has been proven in \cite{Floreani-Redig-Sau} for a large class of models
using orthogonal dualities. It is an open problem instead to characterize the degree of
universality of the limiting functions $f_n$'s within the set of boundary driven
diffusive systems. See the comment on the comparison to Symmetric Exclusion Process that follows below.}

\paragraph{Comparison to product measure.}
Although macroscopically the system satisfies local equilibrium, at microscopic level the stationary state is
substantially different from a product state. For instance, considering the fluctuations of the total number of 
particles $|M|= \sum_{i=1}^N M_i$, equations
\eqref{profile},\eqref{cov} and \eqref{var} give
\be
\label{comparison}
\lim_{N\to\infty} \frac{1}{N} \Big[\mathbb{E}(|M|^2)- \mathbb{E}(|M|)^2\Big]
=
\lim_{N\to\infty} \frac{1}{N} \Big [\mathbb{E}_{loc}(|M|^2)- \mathbb{E}_{loc}(|M|)^2\Big] +\frac{s}{6} (\rho_R-\rho_L)^2
\ee
where $\mathbb{E}_{loc}$ denotes the expectation w.r.t. 
the inhomogeneous product measure with Negative 
Binomials marginals and profile given by \eqref{profile}. 
Of course, the average is the same 
$\mathbb{E}(|M|)  = \mathbb{E}_{loc}(|M|) = \frac{N}{2}(\rho_L+\rho_R)$.
However we find that the fluctuations of $|M|$ in the
non-equilibrium steady state are increased in comparison to those 
of the inhomogeneous product measure.
If one puts $s=-1/2$ then the expression in \eqref{comparison} coincides
with the analogous one for the open symmetric
exclusion process, where fluctuations are decreased in comparison to the product measure (see formula (10) in \cite{Derrida2001}).

\paragraph{Comparison to Symmetric Exclusion Process.}
\cris{Remarkably, we find that 
our microscopic expressions 
{\eqref{profile} and  \eqref{k3}} 
for the first and the third cumulants,
when ``formally'' specialized to the spin value $s=-1/2$, 
do coincide with \cris{the negative of} those of the open symmetric simple exclusion process (cf. \cite{DLS} formula (2.3) and (2.5) with $a=b=1$).
The expression for the covariance \eqref{cov} with $s=-1/2$ is instead 
the same as SEP (cf. \cite{DLS} formula (2.4) with $a=b=1$).
This leads to the conjecture for the cumulants $\kappa_n^{SHM} = (-1)^n \kappa_n^{SEP}$,
where SHM stands for ``Symmetric Harmonic Model''.
This might be related to the fact that the
group behind our model is $SU(1,1)\cong SL(2,\R)$ 
and the group underlying the SEP model is $SU(2)$. }

\paragraph{Absorption probabilities.} The crucial quantities that allow to express the stationary state in closed form
are the absorption probabilities of the dual process
\cite{carinci-giardina-giberti-redig}. For several boundary-driven systems that have an
absorbing dual process, it has been proved  in \cite{2019arXiv190710583C} that the absorption probabilities satisfy a consistency property
implying a recursion relation linking the absorption probabilities of $|\xi|$ dual particles to the absorption probabilities of $|\xi|-1$ particles. However, these recursion relations are not sufficient to determine in closed form 
the absorption probabilities from the single particle absorption probabilities.
In our case, the explicit form of the absorption
probabilities in coordinate form reads
\begin{equation}
p_\xi(k)=
\sum_{m=k}^{|\xi|}(-1)^{n-k}\binom{n}{k}\sum_{1\leq i_1\leq \ldots\leq i_n\leq |\xi|}\prod_{\alpha=1}^n\frac{n-\alpha+2s(N+1-x_{i_\alpha})}{n-\alpha+2s(N+1)}
\end{equation} 
and from this we may indeed verify, as expected, the recursion relations. For instance,
with obvious notation, we have
\begin{equation}
 p_{x_1}(1)+p_{x_2}(1)= 2 p_{x_1,x_2}(2) + p_{x_1,x_2}(1)
\end{equation} 
and
\begin{equation}
 p_{x_2,x_3}(2) + p_{x_1,x_3}(2) + p_{x_1,x_2}(2) =  3 p_{x_1,x_2,x_3}(3) + p_{x_1,x_2,x_3}(2)
\end{equation} 
\begin{equation}
 p_{x_2,x_3}(1) + p_{x_1,x_3}(1) + p_{x_1,x_2}(1) =  2 p_{x_1,x_2,x_3}(2) + 2 p_{x_1,x_2,x_3}(1)\,.
\end{equation}

\paragraph{Proof strategy.}
The key idea behind our proofs
is to use the algebraic description of the 
Markov generator {in terms of 
$\mathfrak{sl}(2)$ Lie algebra generators} (see Section~\ref{sec:algebraic}).
More precisely, the transposed of the process generator
is given by the Hamiltonian $H$ of the open integrable XXX spin chain
with non-compact spins.
The formulation 
within the Quantum Inverse Scattering Method gives full control over the symmetries of the model
and suggests then 
a sequence of similarity transformations
that simplify the spectral problem
of the non-equilibrium Hamiltonian
$H$. The sequence is described by the following
diagram:
\be
\label{eq:seq-trans}
{ H} \;\; 
\stackrel{{\cal S}_{1}}{\longrightarrow} \;\;  { H'} \;\;  
\stackrel{{\cal S}_{2}}{\longrightarrow} \;\;  {H''} \;\;  
\stackrel{{\cal W}^{-1}\,\,\,}{\longrightarrow}  \;\;  {H}^\circ. 
\ee
Here 
$H'$ has boundary terms both triangular, 
$H''$ has the left boundary term triangular and the right boundary term diagonal,
$H^{\circ}$ has boundary terms that are both diagonal.
The arrows in the diagram have to be interpreted as similarity 
transform, e.g. $H'={\cal S}_{1}H{\cal S}_{1}^{-1}$.
As a consequence all the Hamiltonians will be isospectral.
The transformations denoted by ${\cal S}_1$ and ${\cal S}_{2}$ are local, whereas the transformation ${\cal W}$ is non-local.  It will be computed exactly using  perturbation theory following a similar reasoning as presented in \cite{Alcaraz:1992zc}.
For the study of such transformations and the isospectrality for the SEP, we refer the reader to \cite{stinchcombe1995application,2005NuPhB.711..565M,Essler,Crampee,Frassek:2020omo}.

\paragraph{Construction of the steady state.}
From the sequence of transformations  
\eqref{eq:seq-trans} one can construct
the stationary measure of the process.
Indeed the ground state of $H^{\circ}$ is trivial
beacuse the boundary terms are diagonal. It can then  be mapped to the steady state by using backwards the sequence of transformations in \eqref{eq:seq-trans}.
The same strategy has been applied to
the open symmetric exclusion process in  \cite{Frassek:2019imp}, where it provides an alternative route compared to the Matrix Product Ansatz \cite{Derrida_1993}.
The formulation of our solution as a matrix product state is an open problem.

\paragraph{Mapping to equilibrium.}
If one starts from the equilibrium Hamiltonian
$H^{\mathrm{eq}}$ with reservoir parameters 
$\rho_L=\rho_R=\rho$ then the sequence of transformations leading to 
the block diagonal form $H^{\circ}$ reads
\be
\label{eq:seq-trans2}
{ H}^{\mathrm{eq}} \;\; 
\stackrel{{\cal S}_{1}}{\longrightarrow} \;\;  { H'} \;\;  
\stackrel{{\cal S}_{2}}{\longrightarrow} \;\;  {H''} \;\;  
\stackrel{{\cal I}d\,\,\,}{\longrightarrow}  \;\;  {H}^\circ 
\ee
where now 
${\cal W}= {\cal I}d$.
Since the diagonal form $H^{\circ}$ is the same
whether one starts from $H$ or from $H^{\mathrm{eq}}$, 
combining together the two sequences 
\eqref{eq:seq-trans} and \eqref{eq:seq-trans2}
of transformations one gets a similarity 
transformation
\begin{equation}
\label{eq:mapping-intro}
{H}^{\mathrm{eq}}={\cal P}^{-1} \,{H} \,{\cal P}
\end{equation} 
relating the non-equilibrium Hamiltonian
to the equilibrium one.
The operator ${\cal P}$ is given by
\begin{equation}
{\cal P} =  {\cal S}_{1}^{-1}{\cal S}_{2}^{-1}(\rho_R){\cal W}(\rho_L-\rho_R) {\cal S}_{2}(\rho){\cal S}_{1}, 
\end{equation}
where we stressed the dependence of the transformations on the parameter values.
In this paper we focus on the non-equilibrium steady state,
that will then be related to the equilibrium one by 
\be
\mu = {\cal P} \mu^{\mathrm{eq}}\;.
\ee
We expect that the mapping \eqref{eq:mapping-intro} is true
at the level of generator of the process and
therefore, at any time $t\ge 0$, the non-equilibrium
dynamics can be expressed in term of the equilibrium one \cite{Frassek:2020omo}.
The transformation ${\mathcal P}$ relating non-equilibrium
to equilibrium contains both local parts
(the transformations ${\cal S}_{1}$ and ${\cal S}_{2}$)
as well as a non-local part (the transformation ${\cal W}$)
that is responsible for the spatial and temporal 
long-range correlations appearing in non-equilibrium.
For the analogous result at macroscopic
level see the work \cite{Tailleur}, which 
provides the mapping between equilibrium and non-equilibrium 
in the context of Macroscopic Fluctuation 
Theory \cite{jona} via a sequence of non-local change of variables.

\paragraph{Comparison to bulk driven systems.}
An important class of problems 
in non-equilibrium statistical mechanics deals with 
{\em bulk-driven systems}. A prototype example
is the asymmetric exclusion processes on the full line. 
Unlike the case of open boundaries, in these cases 
the system reaches a steady state of 
constant density, and one is interested in density
fluctuations and their correlations. 
Recently, there has been a surge of interest 
in these problems, due to their close relationship 
to growth models, whose continuum version yields 
the celebrated KPZ equation. A whole host of 
asymmetric integrable models have been shown to belong to 
the KPZ universality class  
\cite{sasamoto2010one,corwin2012kardar,borodin,povolotsky2021untangling}.
The bulk-driven version of our model has
been studied by \cite{Sasamoto} for spin $s=1/2$
and by \cite{Povolotsky,barraquand} for general spin $s$.
For a discussion of these bulk driven systems in the framework of the Quantum Inverse Scattering Method we refer the reader to \cite{Kuniba,frassek00}, {see also \cite{frassek22} for a proposal of integrable boundary reservoirs within this setup.}

\section{Algebraic description}
\label{sec:algebraic}
\label{sec:alg-des}
An important role in our analysis is played by the fact that the Markov generator in Definition
\ref{def:process} is related to the integrable
Hamiltonian of the non-compact open Heisenberg spin chain for a particular choice of boundary conditions \cite{Frassek:2019vjt}. 

\subsection{The open non-compact Heisenberg spin chain}

At each lattice site $i\in\{1,2\ldots N\}$ we consider a copy of the $\mathfrak{sl}(2)$ Lie algebra,
whose generators satisfy the commutation relations
\begin{equation}
\label{eq:sl2 discrete}
 [S_0^{[i]},S_\pm^{[i]}]=\pm S_\pm^{[i]}\,,\qquad [S_+^{[i]},S_-^{[i]}]=-2S_0^{[i]}.
\end{equation}
Generators at different sites commute.
The non-compact representations 
that are relevant  are defined by 
\begin{equation}\label{eq:sl2ac}
  S_+^{[i]}|m_i\rangle = (m_i+2s)|m_i+1\rangle\,,\qquad  S_-^{[i]}|m_i\rangle  =  m_i|m_i-1\rangle\,,\qquad S_0^{[i]}|m_i\rangle  =  (m_i+s)|m_i\rangle\,,
\end{equation} 
where $S_+^{[i]}, S_-^{[i]}, S_0^{[i]}$ are infinite-dimensional matrices.
The elements of these matrices are specified by providing the 
action of the matrices on the standard orthonormal base in the Hilbert space $\ell_2(\N_0)$,
i.e. the infinite-dimensional column vectors $|m_i\rangle$, labelled by $m_i=0,1,2,\ldots$, 
with all components zero expect in the $m_i^{th}$ place
where there is a 1. Here the parameter $s$  (spin value) is a positive real number labelling the lowest weight representation.

The quantum Hamiltonian is defined on the tensor product of $N$ irreducible spin $s$ representations and 
it thus acts on the Hilbert space ${\mathscr H}_N = \otimes_{i=1}^N \ell_2(\mathbb{N}_0)$.
It is given by 
\begin{equation}\label{eq:fullham-sss}
H=\mathcal{B}_1+\sum_{i=1}^{N-1}\mathcal{H}_{i,i+1}+\mathcal{B}_N
\end{equation} 
where the nearest neighbour interaction on the tensor product of two sites $i$ and $i+1$ is
\begin{equation}
\label{eq:hacts-ss}
\begin{split}
 \mathcal{H}_{i,i+1}|m_i\rangle\otimes|m_{i+1}\rangle&=\left(h_s(m_i)+h_s(m_{i+1})\right)|m_i\rangle\otimes|m_{i+1}\rangle\\
 &\qquad-\sum_{k=1}^{m_i}
 \implicit{k}{m_i}
 |m_i-k\rangle\otimes|m_{i+1}+k\rangle
 \\&\qquad-\sum_{k=1}^{m_{i+1}}
  \implicit{k}{m_{i+1}}
 |m_i+k\rangle\otimes|m_{i+1}-k\rangle
 \end{split}
\end{equation}
and the boundary terms read
\begin{equation}\label{eq:harmbndss}
  \mathcal{B}_i|m_i\rangle=\left(h_s(m_i)+\sum_{k=1}^\infty\frac{\beta_i^k}{k}\right)|m_i\rangle -\sum_{k=1}^{m_i}
   \implicit{k}{m_i}
|m_i-k\rangle -\sum_{k=1}^\infty \frac{\beta_i^k}{k}|m_i+k\rangle
\end{equation} 
with $\beta_{1} = \beta_L$ and $\beta_N = \beta_R$ and
$h_s$ and $\varphi_s$  defined in \eqref{eq:identity} \eqref{eq:varphi}, respectively. 
With these definitions we immediately have that 
\be
\label{eq:identify}
{L} = - {H}^T
\ee
where ${H}^T$ denotes the transposed of the matrix ${H}$ in \eqref{eq:fullham-sss}, 
and $L$ is matrix associated to the Markov generator ${\mathscr L}$, i.e.
\be
{\mathscr L} f(m) = \sum_{m'\in {\mathscr C}_{N}} L(m,m') f(m').
\ee
We now provide the algebraic description of the Hamiltonian \eqref{eq:fullham-sss}
in terms of the generators $(S_+^{[i]}, S_-^{[i]}, S_0^{[i]})$ with ${i\in\{1,\ldots,N\}}$.
To this aim, we start with two lemmas that will be used several times in what follows. The lemmas are stated for the single site Hilbert space. 
\begin{lemma}[Rotation by $S_-$]
\label{lemma:rotation with S-}
For all $s>0$ and $\alpha\in\C$, the following identity holds with $n\in\N_0$
\be
\label{eq:rotation with S-}
e^{-\alpha S_-}
\Big(\psi(S_0+s)-\psi(2s)\Big) 
e^{\alpha S_-} |n\rangle = 
h_s(n)|n\rangle - \sum_{k=1}^n \implicit{k}{n} \alpha^k |n-k\rangle
\ee
\end{lemma}
\begin{proof}
From the action of the generators \eqref{eq:sl2ac} it follows immediately that
\be\label{eq:rotsm}
e^{\alpha S_-} |n\rangle = \sum_{j=0}^n \alpha^j {n \choose j} |n-j\rangle.
\ee
Furthermore, from the identity \eqref{eq:identity}
one gets
\be
\label{eq:donald}
\Big(\psi(S_0+s)-\psi(2s)\Big) |n\rangle = h_s(n)|n\rangle.
\ee
Using the previous two expressions one obtains
\be
e^{-\alpha S_-}\Big(\psi(S_0+s)-\psi(2s)\Big) e^{\alpha S_-} |n\rangle =
\sum_{j=0}^n\sum_{l=0}^{n-j} (-1)^l \alpha^{j+l} {n \choose j} {n-j \choose l} h_{s}(n-j) |n-j-l\rangle
\ee
We change variable by defining $k=j+l$, thus finding
\be
e^{-\alpha S_-}\Big(\psi(S_0+s)-\psi(2s)\Big) e^{\alpha S_-} |n\rangle =
\sum_{k=0}^n\alpha^{k} \sum_{j=0}^{k} (-1)^{k-j}  {n \choose j} {n-j \choose k-j} h_{s}(n-j) |n-k\rangle
\ee
Formula \eqref{eq:rotation with S-} follows
then by using the identity (valid for all $k\in\{1,\ldots,n\}$)
\be
\label{eq:id-sum1}
\sum_{j=0}^{k} (-1)^{k-j} {n \choose j} {n-j \choose k-j} h_{s}(n-j) = - \implicit{k}{n}.
\ee
It remains to prove \eqref{eq:id-sum1}. This in turn is
obtained by  using the simpler identity 
\begin{equation}
\label{eq:id-sum2}
 \sum_{r=0}^\ell (-1)^r \binom{\ell}{r} \psi(r+b)= -B(\ell,b)\,,\qquad\qquad b>0,
\end{equation} 
where we recall that the Beta function is defined via
\begin{equation}
 B(\ell,b)=\frac{\Gamma(\ell)\Gamma(b)}{\Gamma(\ell+b)}\,.
\end{equation} 
To show \eqref{eq:id-sum2} we employ
a representation of the Beta function
as a sum with alternating signs. 
We use  formula (0.160) from \cite{gradshteyn2007table} 
\begin{equation}\label{eq:GR}
 \sum_{r=0}^\ell (-1)^r \binom{\ell}{r} \frac{\Gamma(r+b)}{\Gamma(r+a)}=\frac{B(\ell+a-b,b)}{\Gamma(a-b)}.
\end{equation} 
Then
\begin{equation}
\begin{split}
 \sum_{r=0}^\ell (-1)^r \binom{\ell}{r} \psi(r+b)&=\lim_{a\to b}\frac{\partial}{\partial {b}}\left( \sum_{r=0}^\ell (-1)^r \binom{\ell}{r} \frac{\Gamma(r+b)}{\Gamma(r+a)}\right)\\
 &=\lim_{a\to b}\frac{\partial}{\partial {b}}\left( \frac{B(\ell+a-b,b)}{\Gamma(a-b)}\right)\\
 &=B(\ell,b)\lim_{a\to b}\frac{ \psi (a-b) }{\Gamma(a-b)}\\
 &=-B(\ell,b)\,.
 \end{split}
\end{equation} 
Here we used  that $\lim_{x\to 0}\frac{\psi(x)}{\Gamma(x)}=-1$, see \cite{psigamma}.
\end{proof}

\begin{lemma}[Rotation by $S_+$]
\label{lemma:rotation with S+}
For all $s>0$ and $\alpha\in\C$, the following identity holds with $n\in\N_0$
\be
\label{eq:rotation with S+}
e^{\alpha S_+}
\Big(\psi(S_0+s)-\psi(2s)\Big) 
e^{-\alpha S_+} |n\rangle = 
h_s(n)|n\rangle - \sum_{k=1}^\infty  \frac{\alpha^k}{k} |n+k\rangle\,.
\ee
\end{lemma}
\begin{proof}
From the action of the generators \eqref{eq:sl2ac} it follows immediately that
\be
e^{\alpha S_+} |n\rangle = \sum_{j=0}^\infty \frac{\alpha^j}{j!} \frac{\Gamma(n+2s+j)}{\Gamma(n+2s)} |n+j\rangle.
\ee
Using this expression and the one in \eqref{eq:donald} one obtains
\be
e^{\alpha S_+}\Big(\psi(S_0+s)-\psi(2s)\Big) e^{-\alpha S_+} |n\rangle =
\sum_{j=0}^\infty\sum_{l=0}^{\infty} (-1)^j \alpha^{j+l} 
\frac{\Gamma(n+2s+j+l)}{j! l! \,\Gamma(n+2s)} h_{s}(n+j) |n+j+l\rangle.
\ee
We change variable by defining $k=j+l$, thus finding
\be
e^{\alpha S_+}\Big(\psi(S_0+s)-\psi(2s)\Big) e^{-\alpha S_+} |n\rangle =
\sum_{k=0}^\infty \alpha^{k} \sum_{j=0}^{k} (-1)^j  
\frac{\Gamma(n+2s+k)}{j! (k-j)! \,\Gamma(n+2s)} h_{s}(n+j) |n+k\rangle.
\ee
Formula \eqref{eq:rotation with S+} follows
by using the identity (valid for all $k\in\N$)
\be
\sum_{j=0}^{k} (-1)^j  
\frac{\Gamma(n+2s+k)}{j! (k-j)! \,\Gamma(n+2s)} h_{s}(n+j) = -\frac{1}{k},
\ee
which in turn is also a consequence of \eqref{eq:id-sum2}.
\end{proof}

\subsection{Bulk Hamiltonian}\label{sec:bulk}
The bulk part of the Hamiltonian \eqref{eq:fullham-sss} consists out of the sum of Hamiltonian densities 
that act non-trivially on the tensor product of two neighboring sites. 
An algebraic expression has been given in \cite{Faddeev:1996iy}, it reads
\begin{equation}\label{eq:fadham}
 \mathcal{H}_{i,i+1}=2\left(\psi(\mathbb{S}_{i,i+1})-\psi(2s)\right)
\end{equation} 
where the operator $\mathbb{S}_{i,i+1}$ is related to the two-site Casimir via $C_{i,i+1}=\mathbb{S}_{i,i+1}(\mathbb{S}_{i,i+1}-1)$, where
\begin{equation}\label{eq:twositeC}
 C_{i,i+1}=2S_0^{[i]}S_0^{[i+1]}- S_+^{[i]}S_-^{[i+1]}-S_-^{[i]}S_+^{[i+1]}+2s(s-1).
\end{equation} 
The algebraic expression for the Hamiltonian density in \eqref{eq:fadham} makes its $\mathfrak{sl}(2)$ symmetry manifest. However it is not convenient when acting on the tensor product basis of two sites, 
i.e. it is not immediately clear that its action is given by \eqref{eq:hacts-ss}. There is a more advantageous algebraic expression for the Hamiltonian density, written explicitly in terms of the generators
at sites $i$ and $i+1$. Furthermore it naturally
consists of two parts associated to the left and a right particle jumps. 
\begin{lemma}[Bulk Hamiltonian] 
\label{lemma:bulk}
The Hamiltonian densities in \eqref{eq:hacts-ss} can algebraically be written as
\begin{equation}\label{eq:hdec}
 \mathcal{H}_{i,i+1}= \mathcal{H}_{i,i+1}^{\rightarrow}+ \mathcal{H}^{\leftarrow}_{i,i+1}
\end{equation} 
where 
\begin{equation}
\label{eq:pippo}
 \mathcal{H}_{i, i+1}^{\rightarrow}=e^{-S_+^{[i+1]}(S_0^{[i+1]}+s)^{-1}S_-^{[i]}}
 \Big(\psi(S_0^{[i]}+s)-\psi(2s)\Big)
 e^{S_+^{[i+1]}(S_0^{[i+1]}+s)^{-1}S_-^{[i]}}
\end{equation} 
and
\begin{equation}
 \mathcal{H}_{i, i+1}^{\leftarrow}=e^{-S_+^{[i]}(S_0^{[i]}+s)^{-1}S_-^{[i+1]}}
 \Big(\psi(S_0^{[i+1]}+s)-\psi(2s)\Big)
 e^{S_+^{[i]}(S_0^{[i]}+s)^{-1}S_-^{[i+1]}}\;.
\end{equation} 
Here obviously we have the symmetry
\begin{equation}
  \mathcal{H}^{\leftarrow}_{i, i+1}= \mathcal{H}^{\rightarrow}_{i+1,i}\;.
\end{equation} 
\end{lemma}
\begin{proof}
It is rather straightforward to convince oneself of  the validity of the splitting of the Hamiltonian density in  \eqref{eq:hdec}. Consider for instance the action of
$\mathcal{H}_{i, i+1}^{\rightarrow}$ in \eqref{eq:pippo}. Using Lemma~\ref{lemma:rotation with S-} we find that 
\begin{equation}
\label{eq:inverse0}
\begin{split}
 \mathcal{H}^{\rightarrow}_{i,i+1}|m_i\rangle\otimes|m_{i+1}\rangle&=h_s(m_i)|m_i\rangle\otimes|m_{i+1}\rangle\\
 &\quad\,-\sum_{k=1}^{m_i}\left(S_+^{[i+1]}\left(S_0^{[i+1]}+s\right)^{-1}\right)^{k}\varphi_s(k,m_i)|m_i-k\rangle\otimes|m_{i+1}\rangle . \end{split}
\end{equation} 
The inverse of $S_0^{[i+1]}+s$ is well defined as $s>0$.
Furthermore, due to the action of the generators
\eqref{eq:sl2 discrete} one has
\be
\label{eq:inverse}\left(S_+^{[i+1]}(S_0^{[i+1]}+s)^{-1}\right)^k |m_{i+1}\rangle= |k+m_{i+1}\rangle \qquad\qquad k\ge 0\,.
\ee
Combining \eqref{eq:inverse0} and \eqref{eq:inverse} one finds
\begin{equation}
 \mathcal{H}^{\rightarrow}_{i,i+1}|m_i\rangle\otimes|m_{i+1}\rangle=h_s(m_i)|m_i\rangle\otimes|m_{i+1}\rangle-\sum_{k=1}^{m_i}\varphi_s(k,m_i)|m_i-k\rangle\otimes|m_{i+1}+k\rangle\,.
\end{equation} 
The left part $\mathcal{H}^{\leftarrow}_{i,i+1}$
is treated similarly and thus the decomposition of the Hamiltonian density  in \eqref{eq:hdec} is verified.
\end{proof}

\begin{remark}
{\em
Introducing pairs of creation and annihilation operators $\oa_i$ and $\oad_i$ satisfying the Heisenberg algebra $[\oa_i,\oad_i]=1$, one
defines the Holstein-Primakoff transformation \cite{Holstein:1940zp}
\begin{equation}
S_0^{[i]}=\oad_i \oa_i+s\,,\qquad
S_-^{[i]}=\oa_i\,
\qquad
S_+^{[i]}=\oad_i(\oad_i\oa_i+2s)\,.
\end{equation}
This allows to write
\begin{equation}\label{eq:oscHam}
\begin{split}
 \mathcal{H}_{i, i+1}^{\rightarrow}=&e^{-\oad_{i+1}\oa_i}\psi\left(\oad_i\oa_i+2s\right)e^{\oad_{i+1}\oa_i}-\psi(2s).
 \end{split}
\end{equation} 
This expression is related to the analysis given in \cite{Derkachov:2005hw}.
As shown there, the fundamental R-matrix of the $\mathfrak{sl}(2)$ chain can be factorises into two parts.  The logarithmic derivative of the R-matrix, which yields the Hamiltonian in the quantum inverse scattering method, thus decomposes as a sum of two corresponding parts. One recovers the left and right Hamiltonian density expressed in terms of Heisenberg pairs as in \eqref{eq:oscHam} when 
considering representation of the Heisenberg algebra
in terms of differential operators acting on polynomials.  

}
\end{remark}

As a consequence of the $\mathfrak{sl}(2)$ symmetry of the Hamiltonian density we have
\begin{equation}\label{eq:comHden}
 [S_{a}^{[i]}+S_{a}^{[i+1]},\mathcal{H}_{i,i+1}]=0 \qquad\qquad \text{for } a\in\{+,-,0\}\,.
\end{equation} 
The symmetry extends to the whole bulk of the chain but is broken by the boundary terms. They are discussed in the next section.
  
\subsection{Boundary Hamiltonian}\label{sec:bnd}
 Let us now turn to the description of the boundaries in the Hamiltonian.  
\begin{lemma}[Boundary Hamiltonian] \label{lem:bnd}
The Hamiltonian boundary terms in \eqref{eq:harmbndss} can algebraically be written as
\begin{equation}\label{eq:bndH}
 \mathcal{B}_i=e^{-S_-^{[i]}}
  e^{\rho_iS_+^{[i]}}
  \Big(\psi(S_0^{[i]}+s)-\psi(2s)\Big)
  e^{-\rho_i S_+^{[i]}}e^{ S_-^{[i]}}\qquad\qquad \text{for } i\in\{1,N\} \,.
\end{equation}
Here we use the variables $\rho_i=\beta_i(1-\beta_i)^{-1}$ defined in \eqref{eq:rho}. 
\end{lemma}
\begin{proof}
To alleviate notation, in this proof, we do not write the site $i$.
The proof is essentially a consequence of Lemma~\ref{lemma:rotation with S-}
and Lemma~\ref{lemma:rotation with S+}.
To obtain \eqref{eq:harmbndss} from the algebraic expression \eqref{eq:bndH} we first evaluate the matrix elements of the operator
\begin{equation}\label{eq:Oop}
 \mathcal{B}'=
  e^{\rho S_+}\psi(S_0+s)e^{-\rho S_+}-\psi(2s)\,.
\end{equation} 
Lemma~\ref{lemma:rotation with S+} yields
\begin{equation}\label{eq:opO}
\begin{split}
 \langle n|\mathcal{B}'|m\rangle
 &=h_s(m)\delta_{n,m} -\frac{\rho^{n-m}}{n-m}\mathds{1}_{\{n>m\}}\,.
 \end{split}
\end{equation} 
With the help of this formula we are in the position to evaluate the matrix elements of the boundary Hamiltonian. Obviously they can be written in terms of \eqref{eq:opO} as
\begin{equation}\label{eq:hamma}
\begin{split}
 \langle p|\mathcal{B}|q\rangle &=\sum_{n,m=0}^\infty \langle p| e^{-S_-}|n\rangle \langle n|\mathcal{B}'|m\rangle\langle m| e^{S_-}|q\rangle=\sum_{n=p}^\infty \sum_{m=0}^q \frac{(-1)^{n-p}}{(q-m)!(n-p)!}\frac{q!n!}{m!p!} \langle n|\mathcal{B}'|m\rangle\,,
 \end{split}
\end{equation} 
where we used the formula \eqref{eq:rotsm}.
To evaluate the right hand side of \eqref{eq:hamma} let us first consider the term that arises from the diagonal contribution of \eqref{eq:opO}. Using Lemma~\ref{lemma:rotation with S-} we find 
\begin{equation}\label{eq:dia}
\begin{split}
\frac{q!}{p!} \sum_{n=p}^\infty \sum_{m=0}^q \frac{(-1)^{n-p}}{(q-m)!(n-p)!}\frac{n!}{m!} h_s(m)\delta_{n,m}
  &=h_s(q)\delta_{q,p}-\varphi_s(q-p,q)\mathds{1}_{\{p<q\}}\,.
 \end{split}
\end{equation} 
Thus the diagonal contribution of \eqref{eq:opO} already gives us half of the terms we are looking for in order to recover \eqref{eq:harmbndss}. The remaining terms arise from the non-diagonal lower triangular part of \eqref{eq:opO}. In agreement with \eqref{eq:harmbndss} we obtain
\begin{equation}\label{eq:offdia}
 -\frac{q!}{p!}\sum_{n=p}^\infty \sum_{m=0}^q \frac{(-1)^{n-p}}{(q-m)!(n-p)!}\frac{n!}{m!}\frac{\rho^{n-m}}{n-m}\mathds{1}_{\{n>m\}}=\delta_{p,q}\sum_{k=1}^\infty\frac{\beta^k}{k}-\frac{\beta^p}{p}\mathds{1}_{\{p>q\}}\,.
\end{equation} 
%
%
To verify  \eqref{eq:offdia} we consider first the case where $p> q$ and rewrite the expression in terms of $\beta$ using a Pfaff transform. We find
\begin{equation}
\begin{split}
 \sum_{n=p}^\infty \sum_{m=0}^q \frac{(-1)^{n-p}}{(q-m)!(n-p)!}\frac{n!q!}{p!m!}\frac{\rho^{n-m}}{n-m}\mathds{1}_{\{n>m\}}
 &=q!\sum_{m=0}^q\frac{ \rho ^{p-m} \, _2F_1(p+1,p-m;-m+p+1;-\rho )}{(p-m) m! (q-m)!}
 \\
&=q!\sum_{m=0}^q\beta^{p-m}\frac{  \, _2F_1(-m,p-m;-m+p+1;\beta )}{(p-m)m!  (q-m)!} \\&=q!\sum_{m=0}^q\sum_{k=0}^m\frac{ (-1)^k \beta^{p+k-m}}{ (p+k-m)k! (m-k)! (q-m)!}
\,.
 \end{split}
\end{equation} 
Further a change of variables $k\to m-k$ and interchanging the sums yields
\begin{equation}
\begin{split}
q!\sum_{m=0}^q\sum_{k=0}^m\frac{ (-1)^k \beta^{p+k-m}}{ (p+k-m)k! (m-k)! (q-m)!}&
=q!\sum_{k=0}^q\sum_{m=k}^q\frac{ (-1)^{m-k} \beta^{p-k}}{ (p-k)(m-k)! k! (q-m)!}\\&=\sum_{k=0}^q\frac{ (-1)^{k} \beta^{p-k}}{ (p-k) k! }\delta_{k,0}\\
&=\frac{\beta^p}{p}\,.
 \end{split}
\end{equation} 
It remains the case $q\geq p$ for which we obtain
\begin{equation}\label{eq:summ}
\begin{split}
 \sum_{n=p}^\infty \sum_{m=0}^q \frac{(-1)^{n-p}}{(q-m)!(n-p)!}\frac{n!q!}{p!m!}\frac{\rho^{n-m}}{n-m}\mathds{1}_{\{n>m\}}&=  \sum_{m=0}^{q}\sum_{n=m+1}^\infty (-1)^{n-p} \binom{q}{m}\binom{n}{p}\frac{\rho^{n-m}}{n-m}\\
 &= \sum_{n=1}^\infty(-1)^{n+p} \frac{\rho^{n}}{n} \sum_{m=0}^{q} (-1)^{m}\binom{q}{m}\binom{n+m}{p} \\
 &= \sum_{n=1}^\infty(-1)^{n+p} \frac{\rho^{n}}{n}\frac{1}{p!}\frac{B(q-p,n+1)}{\Gamma(-p)} \\
 \end{split}
\end{equation} 
To get the second line we shift the boundaries of the sum $n\to m+n$  and simply exchanged the sums. 
Further in the third line we used \eqref{eq:GR} with   $a=n-p+1$ and $b=n+1$. Next we note that the  function  $B(q-p,n+1)$ is finite for $q>p$ but not for $q=p$. In the case $q\to p$ with $p\in \N_0$  we obtain
\begin{equation}
 \lim_{q\to p}\frac{B(q-p,n+1)}{\Gamma(-p)}=(-1)^p p!\,.
\end{equation} 
As a consequence we only get a non-zero contribution for $q=p$ from \eqref{eq:summ}.
We end up with
\begin{equation}
\begin{split}
   \sum_{n=1}^\infty(-1)^{n+p} \frac{\rho^{n}}{n}\frac{1}{p!}\frac{B(q-p,n+1)}{\Gamma(-p)} &
 =\delta_{p,q} \sum_{n=1}^\infty\frac{ (-\rho)^{n}   }{ n}
 =-\delta_{p,q} \sum_{n=1}^\infty\frac{ \beta ^n  }{ n}\,.
 \end{split}
\end{equation} 
This validates the equality in  \eqref{eq:offdia}
and thus concludes the equivalence of the expressions \eqref{eq:harmbndss} and \eqref{eq:bndH} of the boundary Hamiltonian.
%
\end{proof}

\begin{remark}
{\em
The boundary Hamiltonian \eqref{eq:bndH} can be brought to the form that was obtained from the boundary Yang-Baxter equation in \cite{Frassek:2019vjt}. For this we make use of the commutation relations for the $\mathfrak{sl}(2)$ generators
\begin{equation}\label{eq:simtr1}
 e^{\gamma S_\pm}S_\mp e^{-\gamma S_\pm}=S_\mp\mp2\gamma S_0+\gamma^2S_\pm\,,
\end{equation} 
and 
\begin{equation}\label{eq:simtr2}
 e^{\gamma S_\pm}S_0 e^{-\gamma S_\pm}=S_0\mp \gamma S_\pm\,.
\end{equation} 
where $\gamma\in \mathbb{C}$.
It follows that for any function $f$ that only depends on the Cartan element $f(S_0)$ (and not on $S_+$ or $S_-$) we have
\begin{equation}\label{eq:eqive}
  e^{\beta S_+}e^{-\alpha S_-}f(S_0)e^{\alpha S_-} e^{-\beta S_+}=e^{-\frac{\alpha}{1+\alpha\beta}S_-}e^{\beta(1+\alpha\beta)S_+}f(S_0)e^{-\beta(1+\alpha\beta)S_+} e^{\frac{\alpha}{1+\alpha\beta}S_-}\,.
\end{equation} 
The boundary Hamiltonian can then be represented in two equivalent ways
\begin{equation}
\begin{split}
 \mathcal{B}_i&=e^{-S_-^{[i]}}
  e^{\beta_i(1-\beta_i)^{-1} S_+^{[i]}}\psi(S_0^{[i]}+s)e^{-\beta_i(1-\beta_i)^{-1} S_+^{[i]}}e^{ S_-^{[i]}}-\psi(2s)\\
  &=
  e^{\beta_i S_+^{[i]}}e^{-(1-\beta_i)^{-1} S_-^{[i]}}\psi(S_0^{[i]}+s)e^{(1-\beta_i)^{-1}S_-^{[i]}} e^{-\beta_i S_+^{[i]}}-\psi(2s)\,.
  \end{split}
\end{equation} 
Similarly, there is a corresponding rewriting of the boundary  K-matrices that were given in \cite{Frassek:2019vjt}. It can be read off immediately from \eqref{eq:eqive}.
}
\end{remark}

\section{The absorbing dual process}
\label{sec:first step}

In this section we prove the basic properties of Section~\ref{basic}, in particular Proposition~\ref{prop:duality} and Proposition~\ref{prop:facmom-abs}.   
The proof of duality involves a change of representation for the underlying $\mathfrak{sl}(2)$ Lie algebra. We introduce a representation in terms of differential operators and explain how it is intertwined with the representation \eqref{eq:sl2ac}.

\begin{lemma}[Intertwining]
\label{lemma:inter}
Consider the differential operators 
working on differentiable functions
$f: \R \to \R$ as
\begin{equation}
\label{eq:repr-cont} 
\mathscr{S}_0 f = \left(\rho\frac{\partial}{\partial \rho}+s\right)f  \qquad\qquad
\mathscr{S}_- f = \frac{\partial f}{\partial \rho}
\qquad\qquad
\mathscr{S}_+ f = \rho\left(\rho\frac{\partial}{\partial \rho}+2s\right)f 
\end{equation}
They satisfy the commutation relations 
\begin{equation}
\label{eq:comm}
 [\mathscr{S}_0,\mathscr{S}_\pm]=\pm \mathscr{S}_\pm\,,\qquad\qquad 
 [\mathscr{S}_+,\mathscr{S}_-]=-2\mathscr{S}_0\,
\end{equation} 
and therefore they provide a representation of the $\mathfrak{sl}(2) $ Lie algebra.
Define the operator
\begin{equation}
 \mathcal{I}= \sum_{n=0}^\infty  \rho^{n}\langle n|.
\end{equation} 
Then $\mathcal{I}$ intertwines the representation in \eqref{eq:repr-cont} 
and the representation in 
\eqref{eq:sl2ac}.
Namely
\begin{equation}\label{eq:intertw}
\mathscr{S}_0\,\mathcal{I} = \mathcal{I}\,S_0  \,,\qquad \qquad
\mathscr{S}_-\,\mathcal{I} = \mathcal{I}\,S_-\,,\qquad \qquad
\mathscr{S}_+\,\mathcal{I} = \mathcal{I}\,S_+ \;.
\end{equation} 
\end{lemma}
\begin{proof}
The fact that the operators defined in \eqref{eq:repr-cont} do satisfy the $\mathfrak{sl}(2)$ commutation relations in \eqref{eq:comm} can
easily be verified by a direct computation. 
Next, we prove the intertwining relation. 
We have
\begin{equation}
\mathscr{S}_0\,\mathcal{I} = \left(\rho\frac{\partial}{\partial\rho}+s\right)\sum_{n=0}^\infty  \rho^n\langle n|  = \sum_{n=0}^\infty (n+s) \rho^n\langle n| = \mathcal{I}\,S_0
\end{equation} 
where, in the last equality, the left action of $S_0$ on the row vector $\langle n|$ has been used, i.e. $\langle n|S_0  =  (n+s)\langle n|$.
Similarly, we have
\begin{equation}
\mathscr{S}_-\,\mathcal{I} = \frac{\partial}{\partial\rho} \sum_{n=0}^\infty  \rho^n\langle n| = \sum_{n=1}^\infty  n\rho^{n-1}\langle n|= \sum_{n'=0}^\infty  (n'+1)\rho^{n'}\langle n'+1| = \mathcal{I}\,S_-
 \end{equation} 
where now we used that
$\langle n'| S_-   =  (n'+1)\langle n'+1|$
and we made a shift $n=n'+1$.
Finally we also verify that
\begin{equation}
\mathscr{S}_+\,\mathcal{I} =\rho\left(\rho\frac{\partial}{\partial\rho} +2s\right)\sum_{n=0}^\infty  \rho^n\langle n| = \sum_{n=0}^\infty (n+2s) \rho^{n+1}\langle n| = \sum_{n'=1}^\infty (n'+2s-1) \rho^{n'}\langle n'-1|  = \mathcal{I}\,S_+
\end{equation} 
where in the last step it has been used that  $\langle n'|S_+  =  (n'-1+2s)\langle n'-1|$, followed by the shift
$n=n'-1$.
\end{proof}

\bigskip

Using the above lemma we prove the relation between the
original process and the dual process.
\begin{proof}[{\bf Proof of Proposition \ref{prop:duality}}]
As a preliminary step we write the duality function \eqref{eq:duality-fct} algebraically. To this aim, for each site $i$ of the chain, we introduce the diagonal matrix $d^{[i]}$  defined by
\begin{equation}\label{eq:r}
d^{[i]} =  \frac{\Gamma(S_0^{[i]}-s+1)\Gamma(2s)}{\Gamma(S_0^{[i]}+s)}.
\end{equation} 
With this definition the matrix $D$ reads
\be
\label{eq:D-matrix-form}
D =\mathcal{I}^{[0]}(\rho_L) \left( \prod_{i=1}^{N} D^{[i]} \right) \mathcal{I}^{[N+1]}(\rho_R)
\ee
where the product has to be interpreted as a tensor product, and for all $i\in\{1,\ldots,N\}$
\be
\label{eq:bulk-duality-fct}
D^{[i]} =  d^{[i]} e^{S_+^{[i]}} 
\ee
and
\be
\mathcal{I}^{[0]}(\rho_L)= \sum_{m_0=0}^\infty  \rho_L^{m_0}\langle m_{0}|
\qquad\qquad
\mathcal{I}^{[N+1]}(\rho_R)= \sum_{m_{N+1}=0}^\infty  \rho_R^{m_{N+1}}\langle m_{N+1}| .
\ee 
Thus the matrix $D$ acts on the tensor product of $N+2$ sites.
The algebraic expression in \eqref{eq:D-matrix-form} can be verified by using the action of the generators in \eqref{eq:sl2ac}.

Next, since we are dealing with Markov chains taking values on countable state spaces, 
to establish the semigroup duality  \eqref{eq:duality-relation},
it is enough to prove the corresponding duality relation for the Markov generators.
Namely, for all
$m\in {\mathscr C}_{N}$ and $\xi \in {\mathscr C}_{N+2}$
we prove that 
\be
({\mathscr L} D(\cdot,\xi)) (m) = ({\mathscr L}^{\mathrm{dual}} D(m,\cdot)) (\xi)
\ee
or, more explicitly,
\be
\label{eq:h724}
\sum_{m'\in {\mathscr C}_{N}} L(m,m') D(m',\xi) = \sum_{\xi'\in {\mathscr C}_{N+2}} L^{\mathrm{dual}}(\xi,\xi') D(m,\xi').
\ee
This equality says that the matrix $D$ with entries $D(m,\xi)$ given in \eqref{eq:duality-fct}
is the intertwiner between the matrix $L$ describing the action
of the generator ${\mathscr L}$ and the transposed of the matrix $L^{\mathrm{dual}}$
describing the action of the generator 
${\mathscr L^{\mathrm{dual}}}$, i.e.
\be
\label{eq:pluto}
LD = D(L^{\mathrm{dual}})^T.
\ee
In view of the identification \eqref{eq:identify}, the previous equation is equivalent to 
\be
\label{eq:tobeproved}
H^TD = DH^{\mathrm{dual}}.
\ee
Thus our task is to prove \eqref{eq:tobeproved}.
Recalling the algebraic form of the Hamiltonian in \eqref{eq:fullham-sss} we have that
$H^T$ is the sum of $N+1$ terms
\be
H^T=\mathcal{B}_1^T+\sum_{i=1}^{N-1}\mathcal{H}_{i,i+1}^T+\mathcal{B}_N^T.
\ee
Similarly, combining the Definition \ref{def:processdual} of the dual process 
with Lemma \ref{lemma:bulk}, it follows that also $H^{\mathrm{dual}}$
can be written as the sum of $N+1$ terms
\be
\label{eq:ham-dual-def}
H^{\mathrm{dual}}=\mathcal{H}_{0,1}^{\leftarrow}+\sum_{i=1}^{N-1}\mathcal{H}_{i,i+1}+\mathcal{H}_{N,N+1}^{\rightarrow}.
\ee
We shall prove that the intertwining \eqref{eq:tobeproved}  holds because the 
contributions to the transposed and
to the dual Hamiltonian are intertwined 
term by term,
i.e. 
\be
\label{eq:verify1}
\mathcal{B}_1^TD = D \mathcal{H}_{0,1}^{\leftarrow}
\ee
and
\be
\label{eq:verify2}
\mathcal{H}_{i,i+1}^TD = D \mathcal{H}_{i,i+1}
\ee
and
\be
\label{eq:verify3}
\mathcal{B}_N^TD = D \mathcal{H}_{N,N+1}^{\rightarrow}\;.
\ee
One may wonder why the extra sites $0$ and $N+1$ do not appear in the l.h.s. of  \eqref{eq:verify1} and \eqref{eq:verify3}. We remark that
while the operators $\mathcal{B}_i=\mathcal{B}_i(\rho_i)$ can be expanded in positive powers of the densities $\rho_i$, 
the operators $\mathcal{H}_{0,1}^{\leftarrow}$ and $\mathcal{H}_{N,N+1}^{\rightarrow}$ are independent of the densities. The role of the extra site in the $\mathcal{B}_i$ is taken by the densities $\rho_i$. Indeed, the space spanned by the polynomials in $\rho$,  $V=\text{span}(1,\rho,\rho^2,\ldots)$, is isomorphic to the Hilbert space at the extra sites  $V\simeq  \ell_2(\mathbb{N}_0)$. As we shall see below the intertwiner $\mathcal{I}$ provides such isomorphism.

We are now in the position to verify \eqref{eq:verify1}, \eqref{eq:verify2}, \eqref{eq:verify3}.

\bigskip
\noindent
\paragraph{Boundary duality: proof of equations \eqref{eq:verify1} and \eqref{eq:verify3}. }
The proof is essentially the same for the two cases, thus we only prove  \eqref{eq:verify1}.
We start by observing that the following conjugation
relation holds
\be
\label{eq:trans}
d^{[i]} S_{\pm}^{[i]} (d^{[i]})^{-1} = (S_{\mp}^{[i]})^T
\ee
whereas $S_{0}^{[i]} = (S_{0}^{[i]})^T$.
Again, this can be checked using \eqref{eq:sl2ac} and \eqref{eq:r}. We shall now consider separately the left and right hand side of \eqref{eq:verify1}. Starting from the algebraic expression \eqref{eq:bndH} of the left-boundary term
$\mathcal{B}_1$ we find for the transposed
\begin{equation}\label{eq:bndH-trans}
 \mathcal{B}_1^T= d^{[1]} e^{S_+^{[1]}}
  e^{-\rho_L S_-^{[1]}}\Big (\psi(S_0^{[1]}+s)-\psi(2s)\Big)e^{\rho_L S_-^{[1]}}e^{ - S_+^{[1]}} (d^{[1]})^{-1}\,.
\end{equation}
Combining this with \eqref{eq:D-matrix-form}, we arrive to the following expression for the left hand side of \eqref{eq:verify1}
\begin{equation}
\label{eq:left1}
 \mathcal{B}_1^T D = 
D^{[1]}  \left [e^{-\rho_L S_-^{[1]}}\Big (\psi(S_0^{[1]}+s)-\psi(2s)\Big)e^{\rho_L S_-^{[1]}}\mathcal{I}^{[0]} \right] \prod_{i=2}^{N} D^{[i]} \mathcal{I}^{[N+1]} \,.
\end{equation}
Now, it is convenient to interpret the $\rho_L$ in \eqref{eq:left1} as a multiplication operator 
at the extra site $0$ in the representation \eqref{eq:repr-cont}, i.e. for a differentiable function $f:\R\rightarrow\R$
\begin{equation}
 \rho_L f(\rho_L)=\mathscr{S}^{[0]}_+(\mathscr{S}_0^{[0]}+s)^{-1} f(\rho_L) \;.
\end{equation} 
This allows to rewrite  \eqref{eq:left1}
 as
\begin{equation}
\label{eq:left11}
 \mathcal{B}_1^T D = 
D^{[1]}  \left [e^{- \mathscr{S}^{[0]}_+(\mathscr{S}_0^{[0]}+s)^{-1}  S_-^{[1]}}\Big (\psi(S_0^{[1]}+s)-\psi(2s)\Big)e^{ \mathscr{S}^{[0]}_+(\mathscr{S}_0^{[0]}+s)^{-1}  S_-^{[1]}}\mathcal{I}^{[0]} \right] \prod_{i=2}^{N} D^{[i]}\mathcal{I}^{[N+1]} \,.
\end{equation}

We then turn to the right hand side of \eqref{eq:verify1}. Using Lemma~\ref{lemma:bulk} to express the term  $\mathcal{H}_{0,1}^{\leftarrow}$ and combining  with \eqref{eq:D-matrix-form}, we find 
\begin{equation}
\label{eq:left2}
 D \mathcal{H}_{0,1}^{\leftarrow} = D^{[1]}  \left[   \mathcal{I}^{[0]}  e^{-S_+^{[0]}(S_0^{[0]}+s)^{-1}S_-^{[1]}}\Big(\psi\left(S_0^{[1]}+s\right)-\psi(2s)\Big)e^{S_+^{[0]}(S_0^{[0]}+s)^{-1}S_-^{[1]}} \right] \prod_{i=2}^{N} D^{[i]}   \mathcal{I}^{[N+1]} 
\end{equation} 
Now, thanks to Lemma \ref{lemma:inter}, we have  that $\mathcal{I}^{[0]}$ is  the intertwiner between
the generators in the representation \eqref{eq:repr-cont} in terms of differential operators and the generators in the representation \eqref{eq:sl2ac} in terms of the infinite-dimensional matrix.
As a consequence, the terms between square brackets in \eqref{eq:left11} and in \eqref{eq:left2} are the same
and the equation \eqref{eq:verify1}  follows.

\bigskip
\noindent
\paragraph{Bulk duality: proof of equation \eqref{eq:verify2}.}
We start by observing that the elements of the diagonal matrix $d^{[i]}d^{[i+1]}$ coincide with the inverse of the following unnormalized distribution 
\be
\mu^{\mathrm{rev}}(m_i,m_{i+1}) = \frac{\Gamma(m_i+2s)}{m_i! \Gamma(2s)} \frac{\Gamma(m_{i+1}+2s)}{m_{i+1}! \Gamma(2s)}\;.
\ee
This distribution is reversible for the Hamiltonian density ${\cal H}_{i,i+1}$.
In other words, it is immediately verified that the detailed-balance relation 
for the Hamiltonian density ${\cal H}_{i.i+1}$
with respect to the measure $\mu^{\mathrm{rev}}$
holds  and may be written as
\be
\label{eq:det-bal}
{\cal H}_{i,i+1}^T d^{[i]}d^{[i+1]} = d^{[i]}d^{[i+1]} {\cal H}_{i,i+1}.
\ee
Equivalently, this relation may also be verified 
starting from the two-site Casimir \eqref{eq:twositeC}
and using the relations \eqref{eq:trans} to compute its transposed, which then gives
\begin{equation}
 C_{i,i+1}^T=d^{[i]}d^{[i+1]}C_{i,i+1}(d^{[i]})^{-1}(d^{[i+1]})^{-1}\;.
\end{equation} 
Then  \eqref{eq:det-bal} follows from the algebraic form
of the Hamiltonian density given in  \eqref{eq:fadham}.

Considering now the left hand side of \eqref{eq:verify2} and using the algebraic description of the matrix $D$ in \eqref{eq:D-matrix-form} we may write
\be
\label{eq:det-bal2}
\mathcal{H}_{i,i+1}^TD  =  
\mathcal{H}_{i,i+1}^T d^{[i]} d^{[i+1]} e^{S_+^{[i]}} e^{S_+^{[i+1]}} \mathcal{I}^{[0]}\mathcal{I}^{[N+1]} \prod_{j\neq i,i+1}D^{[j]}\,.
\ee
We insert \eqref{eq:det-bal} into \eqref{eq:det-bal2} and we find
\begin{equation}
\mathcal{H}_{i,i+1}^TD 
= 
 d^{[i]} d^{[i+1]} \mathcal{H}_{i,i+1} e^{S_+^{[i]}} e^{S_+^{[i+1]}} \mathcal{I}^{[0]}\mathcal{I}^{[N+1]} \prod_{j\neq i,i+1}D^{[j]}\,.
\end{equation}
Upon using the $\mathfrak{sl}(2)$ invariance of the Hamiltonian density stated in 
\eqref{eq:comHden} we may exchange the Hamiltonian density and the creation operators
thus arriving to
\begin{equation}
\mathcal{H}_{i,i+1}^TD 
=    
 d^{[i]} d^{[i+1]}e^{S_+^{[i]}} e^{S_+^{[i+1]}}  \mathcal{H}_{i,i+1}  \mathcal{I}^{[0]}\mathcal{I}^{[N+1]} \prod_{j\neq i,i+1}D^{[j]}.
\end{equation}
Equation \eqref{eq:verify2} then follows. This concludes the proof of Proposition~\ref{prop:duality}.

\end{proof}

\begin{remark}[Intertwined symmetries]
{\em
From the algebraic description we see that the dual process has two evident  symmetries, namely
\begin{equation}
 \left[H^{\mathrm{dual}},\sum_{k=0}^{N+1}S_a^{[k]}\right]=0\,,\qquad a\in \{0,-\}\,.
\end{equation} 
When $a=0$ this symmetry expresses the particle conservation of the dual process. If $a=-$ then the symmetry is at the origin of the consistency property of the dual process remarked in \cite{2019arXiv190710583C}.
Here we observe that the intertwiner $D$ in \eqref{eq:D-matrix-form} between $H^T$ and $H^{\mathrm{dual}}$ allows to map those symmetries to symmetries of the original process. More precisely, as a consequence of the intertwining relations \eqref{eq:intertw},   the   case $a=0$ leads to 
\begin{equation}\label{eq:sym1}
\left[\rho_L\partial_{\rho_L}+\rho_R\partial_{\rho_R}-S_0^{\mathrm{tot}}+{S_-^{\mathrm{tot}}}, H\right]=0\,,
\end{equation} 
while the case $a=-$ gives
\begin{equation}\label{eq:sym2}
 [\partial_{\rho_L}+\partial_{\rho_R}-S_+^{\mathrm{tot}}+2S_0^{\mathrm{tot}}-S_-^{\mathrm{tot}},H]=0\,.
\end{equation} 
It is noteworthy that a symmetry of the non-equilibrium Hamiltonian arises from the extension of the  Hilbert space related to the dual process.
It would be interesting to understand how symmetries of this type manifest themselves in the study of the spectrum of the process or equivalently of the open spin chain.

}
\end{remark}

\begin{proof}[{\bf Proof of Proposition~\ref{prop:facmom-abs}}]
It is an immediate consequence of
Proposition~\ref{prop:duality}. It follows by choosing $\check\xi = (0,\xi_1,\ldots,\xi_N,0)$ in the duality relation \eqref{eq:duality-relation}
and taking the limit $t\to\infty$. The left hand side of \eqref{eq:duality-relation} will
converge as $t\to\infty$ to the scaled factorial moments of Definition \ref{def:fac}.
The right hand side of \eqref{eq:expansion} is instead a consequence of
the fact that the dual process, as $t\to\infty$,  voids the chain and the configurations
of the form $k \delta_0 + (|\xi|-k) 
\delta_{N+1}$ are absorbing.
\end{proof}

\section{The sequence of similarity transformations}
\label{sec:transform}

In this section we prove  intermediate results that will then be used in Section~\ref{sec:proof-main-th} 
for the proof of the main results. The key idea to compute the absorption probability of the
dual process and thus characterize the stationary state is to implement a sequence of 
similarity transformations to diagonalize the boundaries of the open Hamiltonian. Two of those transformations are local and discussed in Section~\ref{sec:equi}. One transformation is non-local  and given in Section~\ref{app:nonlocal}. The construction of the non-local transformation is based on the identification of a non-trivial symmetry of the open chain obtained from the Quantum Inverse Scattering Method, see Section~\ref{sec:QISM}.  As a by-product of the sequence of similarity transformations a mapping between equilibrium and non-equilibrium is discussed in Section~\ref{sec:equil}.

\subsection{Local transformations}
\label{sec:equi}

Considering the Hamiltonian in \eqref{eq:fullham-sss} we see that, as a consequence of \eqref{eq:comHden}, the bulk part commutes with the sum of the Lie algebra generators 
\begin{equation}\label{eq:bulkinv2}
 \left[S_{a}^{\mathrm{tot}},\sum_{i=1}^{N-1}\mathcal{H}_{i,i+1}\right]=0 \,,
\end{equation} 
where
\begin{equation}\label{eq:totspin}
S_{a}^{\mathrm{tot}}:= \sum_{i=1}^N S_{a}^{[i]} \qquad\qquad \text{for } a\in\{+,-,0\}\,.
\end{equation} 
However, this is not the case for the boundary terms ${\cal B}_1$ and ${\cal B}_N$ in the Hamiltonian \eqref{eq:fullham-sss}. 
It follows that local transformations of the form $f(S_{+}^{\mathrm{tot}},S_{-}^{\mathrm{tot}},S_{0}^{\mathrm{tot}})$ will only affect the boundaries and thus can be used to simplify them. 

\begin{proposition}[Local transformation]
\label{prop:local}

\phantom{2}
\begin{enumerate}
 \item[i)]
Define the Hamiltonian $H'$
\begin{equation}
\label{eq:h prime transform}
{H}':=e^{S_-^{\mathrm{tot}}} {H} e^{-S_-^{\mathrm{tot}}}\,.
\end{equation} 
where $S_-^{\mathrm{tot}}$ is given in \eqref{eq:totspin}.
Then 
\begin{equation}\label{eq:Hp}
 {H}'=\mathcal{B}_1'+\sum_{i=1}^{N-1}\mathcal{H}_{i,i+1}+\mathcal{B}_N'
\end{equation} 
with boundary terms
\begin{align}
  \mathcal{B}_1' &=
  e^{\rho_L S_+^{[1]}}\Big(\psi(S_0^{[1]}+s)-\psi(2s)\Big)e^{-\rho_L S_+^{[1]}}
\\
  \mathcal{B}_N' &=
  e^{\rho_R S_+^{[N]}}\Big(\psi(S_0^{[N]}+s)-\psi(2s)\Big)e^{-\rho_R S_+^{[N]}}\,.
 \end{align}
 
 \item[ii)]
Define
\begin{equation}
{H}'':=e^{-\rho_R S_+^{\mathrm{tot}}}{H}'e^{\rho_R S_+^{\mathrm{tot}}}\,.
\end{equation} 
Then 
\begin{equation}\label{eq:Hpp}
{H}''=\mathcal{B}_1''+\sum_{i=1}^{N-1}\mathcal{H}_{i,i+1}+\mathcal{B}_N''
\end{equation} 
with boundary terms
\begin{align}
  \mathcal{B}_1'' &=
  e^{(\rho_L-\rho_R) S_+^{[1]}}\Big(\psi(S_0^{[1]}+s)-\psi(2s)\Big)e^{-(\rho_L-\rho_R) S_+^{[1]}}
\\
  \mathcal{B}_N'' &=
 \psi(S_0^{[N]}+s)-\psi(2s)\,.
 \end{align}

\end{enumerate}

\end{proposition}

\begin{proof}
Recalling the Hamiltonian \eqref{eq:fullham-sss},
the proof of item i) is an immediate consequence of Lemma~\ref{lem:bnd} which provides the algebraic description of the boundary terms of the Hamiltonian and of the $\mathfrak{sl}(2)$ invariance of the bulk part of the Hamiltonian, cf.~\eqref{eq:bulkinv2}. Item ii) follows  from item i).
\end{proof}

\begin{remark}
 {\em
 The Hamiltonian $H'$ is closely related to the Hamiltonian of the  dual process $H^{\mathrm{dual}}$ defined in \eqref{eq:ham-dual-def}. More precisely one has 
\begin{equation}
 (H')^T \mathcal{I}^{0}\mathcal{I}^{N+1}\prod_{i=1}^Nd^{[i]} =\mathcal{I}^{0}\mathcal{I}^{N+1}\prod_{i=1}^N d^{[i]} H^{\mathrm{dual}}
\end{equation} 
This equation itself has the form of a duality relation, compare~\eqref{eq:tobeproved}. However the Hamiltonian $H'$ does not describe a stochastic process.
 }
\end{remark}
We also observe that the boundaries of $H'$ contain only the creation  and the number operator and therefore are  upper-triangular in our basis. In the  Hamiltonian $H''$ one of the two boundaries is diagonal.  Interestingly, there exists yet another similarity transformation $\mathcal{W}$, which is non-local,  that also diagonalises the other boundary. This transformation is obtained in Section~\ref{app:nonlocal} from a non-local charge of the Hamiltonian $H''$ that is discussed in the next section.

\subsection{Hidden symmetries: the non-local charges.}
\label{sec:QISM}

\begin{proposition}[Non-local charge $Q''$]
\label{prop:symm}
The operator $Q''$ defined by
\begin{equation}\label{eq:Qprime}
 Q''=Q^{\circ} -(\rho_L-\rho_R)Q_+
\end{equation}
where
\be
Q^{\circ} = S_0^{\mathrm{tot}}\left(S_0^{\mathrm{tot}}+2s-1\right)
\ee 
and
\begin{equation}\label{eq:Qplus}
\begin{split}
 Q_+&=sS_+^{\mathrm{tot}}+\sum_{i=1}^N S_+^{[i]}\left(S_0^{[i]}+2\sum_{j=i+1}^NS_0^{[j]}\right)\\
 \end{split}
\end{equation} 
commutes with the Hamiltonian $H''$, i.e.
\be
[H'',Q''] =0\;.
\ee
Here $S^{\mathrm{tot}}_a$ with $a\in\{0,+,-\}$ are the total spin operators in \eqref{eq:totspin}.
\end{proposition}

\begin{remark}[Origin of the symmetry $Q''$]
{\em 
The symmetry $Q''$ of the Hamiltonian $H''$ in \eqref{eq:Hpp} is determined within the framework of the quantum inverse scattering method \cite{Faddeev:1996iy,Sklyanin:1988yz}.
The underlying transfer matrix for the Hamiltonian $H$ was given in \cite{Frassek:2019vjt}. 
Applying the similarity transformations discussed in the previous subsection one obtains the transfer matrix $T(z)$ with spectral parameter $z\in\C$ that commutes with the Hamiltonian $H''$
\begin{equation}\label{eq:HTcom}
 [H'',T(z)]=0\,.
\end{equation} 
It is defined explicitly by the following trace over the $2\times2$ auxiliary space
\begin{equation}\label{eq:transm}
 T(z)=\tr K(z)L_1(z)\cdots L_N(z)\hat K(z)L_N(z)\cdots L_1(z)\,.
\end{equation} 
Here the  Lax matrix $L_i(z)$ at site $i$ is defined via
\begin{equation}
L_i(z)=\left(\begin{array}{cc}
             z+\frac{1}{2}+S_0^{[i]}&-S_-^{[i]}\\
             S_+^{[i]}& z+\frac{1}{2}-S_0^{[i]}
            \end{array}
 \right)   \,,
\end{equation} 
and  the K-matrix for the left and right boundary read
\begin{equation}
  K(z)=\left(\begin{array}{cc}
             p_1+(z+1)p_2&-(z+1)\Delta\\
             0& p_1-(z+1)p_2
            \end{array}
\right)\,,\qquad  {\hat K}(z)=\left(\begin{array}{cc}
             q_1+zq_2&0\\
             0&q_1-zq_2
            \end{array}
\right)\,.
\end{equation}
The boundary parameters are fixed via 
\begin{equation}
 q_1=p_1=\frac{1}{2}\left(s-\frac{1}{2}\right)\,,\qquad q_2=p_2=\frac{1}{2}\,,\qquad \Delta=\rho_L-\rho_R\,,
\end{equation} 
see \cite{Frassek:2019vjt}.

The transfer matrix \eqref{eq:transm} is a polynomial of order $2N+2$  in the spectral parameter $z$.  Thus, as a consequence of \eqref{eq:HTcom}, all its coefficients $Q_k$ at order  $z^k$, $k=0,1,2,\ldots,2N+2$, commute with the Hamiltonian $H''$.
The first two charges $Q_{2N+2}$ and $Q_{2N+1}$ are diagonal and independent of the densities $\rho_L$ and $\rho_R$. The first charge that is non-diagonal and dependent on $\Delta=\rho_L-\rho_R$ appears at the order $z^{2N}$. It reads
\begin{equation}
 Q_{2N}=\frac{1}{8} \left(2 N^2+N (3-4 (s-1) s)+(1-2 s)^2\right)\ID+{Q''}\,,
\end{equation} 
with $Q''$ given in \eqref{eq:Qprime}.
We neglect the term proportional to the identity $\ID$ as it trivially commutes with the Hamiltonian.
}
\end{remark}

\noindent
\begin{proof}[{\bf Proof of Proposition~\ref{prop:symm}}]

For the sake of convenience, we abbreviate $\Delta = \rho_L -\rho_R$.
We would like to show that the Hamiltonian $H''$ commutes with the non-local charge $Q''$, i.e.
\begin{equation}\label{eq:HQ}
 [H'',Q'']=0\,.
\end{equation} 
We recall the expression for $H''$ in \eqref{eq:Hpp} from which we see that its right boundary $\mathcal{B}''_N$ is diagonal whereas its left boundary $\mathcal{B}''_1$ is triangular. It is convenient then to decompose the Hamiltonian $H''$  as 
\begin{equation}\label{eq:Hppdec}
 H''=H^\circ+\sum_{k=1}^\infty \Delta^k B_k^{[1]}
\end{equation} 
where $B_k^{[1]}$ acts on the first site only and is upper triangular (with zeros on the diagonal), while $H^\circ$ is
the Hamiltonian with  diagonal boundaries (which  is independent of the parameter $\Delta$). It is given by
\begin{equation}
\label{eq:H0-def}
 H^\circ=\mathcal{B}_1^\circ+\sum_{i=1}^{N-1}\mathcal{H}_{i,i+1}+  \mathcal{B}_N^\circ
\end{equation} 
with the boundary terms
\begin{equation}
 \mathcal{B}_i^\circ= \psi(S_0^{[i]}+s)-\psi(2s)\,,\qquad i=1,N\,.
\end{equation} 
The commutator in \eqref{eq:HQ} can then be expanded in terms of $\Delta $ using \eqref{eq:Qprime} and \eqref{eq:Hppdec}. We get
\begin{equation}
 [H'',Q'']=[H^\circ,Q^{\circ}]-\Delta  \left([Q^{\circ},B_1^{[1]}]-[Q_+,H^\circ]\right)-\sum_{k=2}^\infty \Delta ^k \left([Q^{\circ},B_k^{[1]}]-[Q_+,B_{k-1}^{[1]}]\right)\,.
\end{equation} 
The first term $[H^\circ,Q^{\circ}]$ vanishes due to particle conservation. To show that the other terms vanish some more work is required. 
While the term linear in $\Delta$ is rather involved, the higher order terms can be shown to vanish more easily. 
Thus as a preliminary step we identify the operators $B_k^{[1]}$ and show that the higher order terms vanish. In a second step we will consider the linear term.

\paragraph{The operators $B_k^{[1]}$.}
The operators $B_k^{[1]}$ in \eqref{eq:Hppdec} only act non-trivially on the first site. They are obtained from the expansion of the boundary term
\begin{equation}
\begin{split}
  \mathcal{B}_1''&=
  e^{\Delta  S_+^{[1]}}\psi\left(S_0^{[1]}+s\right)e^{-\Delta  S_+^{[1]}}-\psi(2s)\\
  &=\sum_{n,m=0}^\infty \frac{(-1)^m}{n!m!}\Delta ^{n+m}\left(S_+^{[1]}\right)^{n+m}\psi\left(S_0^{[1]}+s+m\right)-\psi(2s)\\
  &=\sum_{n=0}^\infty\sum_{l=n}^\infty \frac{(-1)^{l-n}}{n!(l-n)!}\Delta ^l \left(S_+^{[1]}\right)^l \psi\left(S_0^{[1]}+s+l-n\right)-\psi(2s)\\
  &=\psi\left(S_0^{[1]}+s\right)-\psi(2s)+\sum_{l=1}^\infty\Delta ^l \left(S_+^{[1]}\right)^l\sum_{n=0}^l \frac{(-1)^{l-n}}{n!(l-n)!} \psi\left(S_0^{[1]}+s+l-n\right)
  \end{split}
\end{equation} 
where in the last step we exchanged the two sums.
Here the first part is the term $\mathcal{B}_1^\circ$ contained in $H^\circ$. The last part can be written using \eqref{eq:id-sum2} as
\begin{equation}
 \sum_{n=0}^l \frac{(-1)^{l-n}}{n!(l-n)!} \psi\left(S_0^{[1]}+s+l-n\right)
 =-\frac{1}{l}\frac{\Gamma\left(S_0^{[1]}+s\right)}{\Gamma\left(S_0^{[1]}+s+l\right)}.
\end{equation} 
Thus we get
\begin{equation}\label{eq:Bk}
 B_k^{[1]}= -\frac{1}{k}\frac{\Gamma(S_0^{[1]}+s-k)}{\Gamma(S_0^{[1]}+s)}\left(S_+^{[1]}\right)^k.
\end{equation}

\paragraph{The higher order terms.}
Let us first show that 
 \begin{equation}
  [Q^{\circ},B_k^{[1]}]-[Q_+,B_{k-1}^{[1]}]=0.
 \end{equation} 
The first commutator is  computed using \eqref{eq:Bk}. We have
\begin{equation}\label{eq:Q0Bk}
\begin{split}
 [Q^{\circ},B_k^{[1]}]
 &=k (2 S_0^{\mathrm{tot}}-k+2 s-1)B_k^{[1]}\,.
\end{split}
 \end{equation} 
 In the  second commutator only the first term in the sum \eqref{eq:Qplus} defining $Q_+$ contributes. We find 
 \begin{equation}
 \begin{split}
  [Q_+,B_{k-1}^{[1]}]&=
\Big[sS_+^{[1]}+  S_+^{[1]}\Big(S_0^{[1]}+2\sum_{j=2}^NS_0^{[j]}\Big),B_{k-1}^{[1]}\Big]\\
&=
s\Big[S_+^{[1]},B_{k-1}^{[1]}\Big]+
\Big[  S_+^{[1]}S_0^{[1]},B_{k-1}^{[1]}\Big]+
2\Big[  S_+^{[1]},B_{k-1}^{[1]}\Big]\sum_{j=2}^NS_0^{[j]}\\
&=\Big[S_+^{[1]},B_{k-1}^{[1]}\Big](s+2S_0^{\mathrm{tot}}-S_0^{[1]})+S_+^{[1]}\Big[S_0^{[1]},B_{k-1}^{[1]}\Big]\,.
  \end{split}
 \end{equation} 
 Using  again \eqref{eq:Bk} we get
 \begin{equation}
 [S_+^{[1]},B_{k-1}^{[1]}]=kB_{k}^{[1]}
\end{equation} 
and
\begin{equation}
 S_+^{[1]}[S_0^{[1]},B_{k-1}^{[1]}]=(k-1) S_+^{[1]} B_{k-1}^{[1]}=k(S_0^{[1]} +s-1)B_{k}^{[1]}\,.
\end{equation}
It then follows that 
 \begin{equation}
 \label{eq:Q+B}
 \begin{split}
  [Q_+,B_{k-1}^{[1]}]&=kB_{k}^{[1]}(s+2S_0^{\mathrm{tot}}-S_0^{[1]})+k(S_0^{[1]} +s-1)B_{k}^{[1]}\\
  &=k(2S_0^{\mathrm{tot}}-k+2s-1)B_{k}^{[1]}\\
  \end{split}
 \end{equation}
 and thus combining \eqref{eq:Q0Bk} with \eqref{eq:Q+B} we obtain the desired result
 \begin{equation}
  [Q^{\circ},B_k^{[1]}]-[Q_+,B_{k-1}^{[1]}]=0\,.
 \end{equation} 

\paragraph{Linear order term.}
Let us now show that 
\begin{equation}\label{eq:com11} 
 [Q^{\circ},B_1^{[1]}]-[Q_+,H^\circ]=0\,.
\end{equation} 
This is more cumbersome as it contains the Hamiltonian which is more involved. The first term we get from \eqref{eq:Q0Bk} by specialising to $k=1$:
\begin{equation}\label{eq:Q0B1}
\begin{split}
 [Q^{\circ},B_1^{[1]}]
 &=2(S_0^{\mathrm{tot}}+ s-1)B_1^{[1]}=-2S_+^{[1]}(S_0^{\mathrm{tot}}+ s)\frac{1}{S_0^{[1]}+s}\,.
\end{split}
 \end{equation} 
 We now turn to the second term. First, using the Definition \eqref{eq:H0-def} of $H^\circ$,  we note that it  splits as
 \begin{equation}
  [Q_+,H^\circ]=
  [Q_+,\mathcal{B}_1^\circ]+\sum_{k=1}^{N-1}
  [Q_+,\mathcal{H}_{k,k+1}]+
  [Q_+,\mathcal{B}_N^\circ]\,.
 \end{equation} 
 The commutators with the boundary terms are computed as follows.  
 Again, only the first term in the sum defining $Q_+$ contributes and thus
 we have
 \begin{equation}\label{eq:QpB10}
  [Q_+,\mathcal{B}_1^\circ]=[sS_+^{[1]}+  S_+^{[1]}\left(S_0^{[1]}+2\sum_{j=2}^NS_0^{[j]}\right),\mathcal{B}_1^\circ]=[ S_+^{[1]},\mathcal{B}_1^\circ]\left(s+S_0^{[1]}+2\sum_{j=2}^NS_0^{[j]}\right)\,.
 \end{equation} 
 The last commutator is evaluated using the recurrence property of the digamma function $\psi(k+1)-\psi(k)=k^{-1}$, so that
 \begin{equation}\label{eq:SpB10}
 \begin{split}
  [ S_+^{[1]},\mathcal{B}_1^\circ]&=  [ S_+^{[1]},\psi(S_0^{[1]}+s)]\\
  &=  S_+^{[1]} \left(\psi(S_0^{[1]}+s)-\psi(S_0^{[1]}+1+s)\right)\\
& =-S_+^{[1]}\frac{1}{S_0^{[1]}+s}
\,.
  \end{split}
 \end{equation} 
 Combining \eqref{eq:Q0B1},  \eqref{eq:QpB10} and \eqref{eq:SpB10} we find the first contribution to \eqref{eq:com11}. It simply reads
 \begin{equation}
 \begin{split}
   [Q^{\circ},B_1^{[1]}]-  [Q_+,\mathcal{B}_1^\circ]
&=-2S_+^{[1]}(S_0^{\mathrm{tot}}+ s)\frac{1}{S_0^{[1]}+s}+S_+^{[1]}\frac{1}{S_0^{[1]}+s}\left(s-S_0^{[1]}+2S_0^{\mathrm{tot}}\right)\\
   &=-S_+^{[1]}\,.
   \end{split}
 \end{equation} 
 The  boundary contribution at the site $N$ yields
  \begin{equation}
  [Q_+,\mathcal{B}_N^\circ]
  =[ S_+^{[N]},\mathcal{B}_N^\circ]\left(S_0^{[N]}+s\right)=-S_+^{[N]}\,,
 \end{equation} 
 where we used the analogous of \eqref{eq:SpB10}, now for site $N$.
 
 We turn to the computation of the commutation relations involving the  Hamiltonian density. 
 To show that the commutator in \eqref{eq:com11} vanishes it remains to prove that 
 \begin{equation}
  \label{eq:Q+expan}
  \sum_{k=1}^{N-1}
  [Q_+,\mathcal{H}_{k,k+1}]=-S_+^{[1]}+S_+^{[N]}\,.
 \end{equation} 
 We verify below this equality  through a direct computation where we act on states of our Hilbert space.
 First note that the operator $Q_+$ can be rewritten as
\begin{equation}
\begin{split}
 Q_+
 &=S_+^{\mathrm{tot}}(s+S_0^{\mathrm{tot}})-\sum_{i=1}^NS_+^{[i]} \left(\sum_{j=1}^{i-1}S_0^{[j]}-\sum_{j=i+1}^NS_0^{[j]}\right)\,.
 \end{split}
\end{equation} 
Thus the commutator on the left hand side of \eqref{eq:Q+expan} can be expanded as
 \begin{equation}
 \label{eq:Q+Hbulkkk}
 \begin{split}
   [Q_+,\mathcal{H}_{k,k+1}]
   &= \sum_{i=1}^N\left[S_+^{[i]} ,\mathcal{H}_{k,k+1}\right]\sum_{j=i+1}^NS_0^{[j]}
   -\sum_{i=1}^N\left[S_+^{[i]} ,\mathcal{H}_{k,k+1}\right] \sum_{j=1}^{i-1}S_0^{[j]}\\
   &
   \quad\;
   +\sum_{i=1}^NS_+^{[i]} \left[ \sum_{j=i+1}^NS_0^{[j]},\mathcal{H}_{k,k+1}\right]
      -\sum_{i=1}^NS_+^{[i]}\left[ \sum_{j=1}^{i-1}S_0^{[j]},\mathcal{H}_{k,k+1}\right].    
   \end{split}
 \end{equation} 
 Let us first simplify the second line of the  equation above.
 By using the $\mathfrak{sl}(2)$ invariance of the Hamiltonian density (see \eqref{eq:comHden})
 we obtain
 \begin{equation}
   \left[ \sum_{j=i+1}^NS_0^{[j]},\mathcal{H}_{k,k+1}\right]=-\delta_{i,k}
   \left[ S_0^{[k]},\mathcal{H}_{k,k+1}\right]
 \end{equation} 
  and
 \begin{equation}
   \left[ \sum_{j=1}^{i-1}S_0^{[j]},\mathcal{H}_{k,k+1}\right]=-\delta_{i,k+1}
   \left[ S_0^{[k+1]},\mathcal{H}_{k,k+1}\right]
 \end{equation} 
 where $\delta_{i,k}$ denotes the Kronecker delta.
Therefore we find
 \begin{equation}
 \begin{split}
  \sum_{i=1}^NS_+^{[i]} \left[ \sum_{j=i+1}^NS_0^{[j]},\mathcal{H}_{k,k+1}\right]  -\sum_{i=1}^NS_+^{[i]}\left[ \sum_{j=1}^{i-1}S_0^{[j]},\mathcal{H}_{k,k+1}\right]    &=-S_+^{k}[S_0^{[k]},\mathcal{H}_{k,k+1}]+S_+^{k+1}[S_0^{[k+1]},\mathcal{H}_{k,k+1}]\\
  &=-(S_+^{[k]}+S_+^{[k+1]})[S_0^{[k]},\mathcal{H}_{k,k+1}].
  \end{split}
 \end{equation} 
 For the first line of \eqref{eq:Q+Hbulkkk} we have
  \begin{equation}
 \begin{split}
&\sum_{i=1}^N\left[S_+^{[i]} ,\mathcal{H}_{k,k+1}\right] \sum_{j=i+1}^NS_0^{[j]}-\sum_{i=1}^N\left[S_+^{[i]} ,\mathcal{H}_{k,k+1}\right] \sum_{j=1}^{i-1}S_0^{[j]}   \\&=\left[S_+^{[k]} ,\mathcal{H}_{k,k+1}\right] \left(\sum_{j=k+1}^NS_0^{[j]}-\sum_{j=1}^{k-1}S_0^{[j]}   \right)  +\left[S_+^{[k+1]} ,\mathcal{H}_{k,k+1}\right] \left(\sum_{j=k+2}^NS_0^{[j]}-\sum_{j=1}^{k}S_0^{[j]}   \right)  \\
&=\left[S_+^{[k]} ,\mathcal{H}_{k,k+1}\right] \left(\sum_{j=k+1}^NS_0^{[j]}-\sum_{j=1}^{k-1}S_0^{[j]}  -\sum_{j=k+2}^NS_0^{[j]}+\sum_{j=1}^{k}S_0^{[j]}  \right)  \\
&=\left[S_+^{[k]} ,\mathcal{H}_{k,k+1}\right] \left(S_0^{[k+1]}+S_0^{[k]}  \right).\\
\end{split}
 \end{equation} 
 In the first equality we used that only the terms with $i=k,k+1$ contribute to the commutators and the second equality uses again the $\mathfrak{sl}(2)$ invariance of the Hamiltonian density.
It then follows that we can write \eqref{eq:Q+Hbulkkk} as
 \begin{equation}
 \begin{split}
   [Q_+,\mathcal{H}_{k,k+1} ]&=\left[S_+^{[k]} ,\mathcal{H}_{k,k+1}\right] \left(S_0^{[k+1]}+S_0^{[k]}  \right)-(S_+^{[k]}+S_+^{[k+1]})[S_0^{[k]},\mathcal{H}_{k,k+1}]\\
   &=\left[S_+^{[k]} \left(S_0^{[k+1]}+S_0^{[k]}  \right),\mathcal{H}_{k,k+1}\right] -[(S_+^{[k]}+S_+^{[k+1]})S_0^{[k]},\mathcal{H}_{k,k+1}]\\
   &=\left[S_+^{[k]} S_0^{[k+1]},\mathcal{H}_{k,k+1}\right] -[S_+^{[k+1]}S_0^{[k]},\mathcal{H}_{k,k+1}].\\
   \end{split}
 \end{equation} 
The proof of \eqref{eq:Q+expan} that we are after will be completed by showing that
\begin{equation}
\label{eq:together}
 \left[S_+^{[k]} S_0^{[k+1]},\mathcal{H}_{k,k+1}\right] -[S_+^{[k+1]}S_0^{[k]},\mathcal{H}_{k,k+1}]=-S_+^{[k]}+S_+^{[k+1]}.
\end{equation} 
This equation is proved hereafter by acting on the tensor product of two arbitrary states.
For convenience we evaluate the two commutators in the left hand side of \eqref{eq:together} separately.
 \paragraph{The commutator $\left[S_+^{[k]} S_0^{[k+1]},\mathcal{H}_{k,k+1}\right]$.}
 By recalling the action of the Hamiltonian density \eqref{eq:hacts-ss}, the first commutator is explicitly written as
  \begin{equation}
 \begin{split}
  &\left[S_+^{[i]}S_0^{[i+1]} ,\mathcal{H}_{i,i+1}\right]|m_i\rangle\otimes|m_{i+1}\rangle\\
  &=(m_i+2s)(m_{i+1}+s)\left(h_s(m_i)+h_s(m_{i+1})\right)|m_i+1\rangle\otimes|m_{i+1}\rangle\\
 &\qquad-\sum_{k=1}^{m_i}(m_i-k+2s)(m_{i+1}+k+s)
 \implicit{k}{m_i}
 |m_i+1-k\rangle\otimes|m_{i+1}+k\rangle
 \\&\qquad-\sum_{k=1}^{m_{i+1}}(m_i+k+2s)(m_{i+1}-k+s)
  \implicit{k}{m_{i+1}}
 |m_i+1+k\rangle\otimes|m_{i+1}-k\rangle\\
 &\qquad -(m_i+2s)(m_{i+1}+s)\mathcal{H}_{i,i+1}|m_i+1\rangle\otimes|m_{i+1}\rangle
 \end{split}
 \end{equation} 
 where
 \begin{equation}
\begin{split}
 \mathcal{H}_{i,i+1}|m_i+1\rangle\otimes|m_{i+1}\rangle&=\left(h_s(m_i+1)+h_s(m_{i+1})\right)|m_i+1\rangle\otimes|m_{i+1}\rangle\\
 &\qquad-\sum_{k=1}^{m_i}
 \implicit{k}{m_i+1}
 |m_i+1-k\rangle\otimes|m_{i+1}+k\rangle
 \\&\qquad-\sum_{k=1}^{m_{i+1}}
  \implicit{k}{m_{i+1}}
 |m_i+1+k\rangle\otimes|m_{i+1}-k\rangle\\
 &\qquad-
 \implicit{m_i+1}{m_i+1}
 |0\rangle\otimes|m_{i+1}+m_i+1\rangle.
 \end{split}
\end{equation} 
It is convenient to add up the coefficients of the base vectors. For the vector $|m_i+1\rangle\otimes|m_{i+1}\rangle$
we have that the global coefficient simplifies to
\be
(m_i+2s)(m_{i+1}+s)\left[h_s(m_i)+h_s(m_{i+1})-h_s(m_{i}+1)-h_s(m_{i+1})\right] = - (m_{i+1}+s)
\ee
where we have used that
\be
h_s(m_i+1) = h_s(m_i) + \frac{1}{m_i+2s}
\ee
which in turn follows from \eqref{eq:identity} using the recurrence property of the digamma
function.
For the vector $|m_i+1-k\rangle\otimes|m_{i+1}+k\rangle$,
when we group together its coefficients
we find 
\begin{equation}
\begin{split}
\label{eq:chi}
 -(m_i-k+2s)(m_{i+1}+k+s)
 \implicit{k}{m_i}&+(m_i+2s)(m_{i+1}+s) \implicit{k}{m_i+1}\\
 &=\frac{k (m_i+2 s) (k-m_i+m_{i+1}+s-1)}{m_i+1}\implicit{k}{m_i+1}\\
 & = (k-1) (k-m_i+m_{i+1}+s-1) \implicit{k-1}{m_i} \\
 & =: \chi(m_i,m_{i+1},k).
\end{split}
\end{equation} 
In the simplifications above we have used that
(cf. \eqref{eq:varphi})
\be
\implicit{k}{m_i+1} = \frac{(m_i+1)(m_i-k+2s)}{(m_i+2s)(m_i-k+1)}
\implicit{k}{m_i}
\ee
and
\be
\implicit{k}{m_i} = \frac{k-1}{k}\frac{m_i-k+1}{m_i-k+2s}
\implicit{k-1}{m_i}\,.
\ee
Finally, for the vector $|m_i+1+k\rangle\otimes|m_{i+1}-k\rangle$,
we find a global coefficient
\begin{equation}
\begin{split}
 -(m_i+k+2s)(m_{i+1}-k+s)\implicit{k}{m_{i+1}}&+(m_i+2s)(m_{i+1}+s)\implicit{k}{m_{i+1}}\\
 & =k(k+m_i-m_{i+1}+s)\implicit{k}{m_{i+1}}\\
& = \chi(m_{i+1},m_i,k-1) 
 \end{split}
\end{equation} 
where we used the definition of $\chi$ in \eqref{eq:chi}.
We thus arrive to the expression
  \begin{equation}
 \begin{split}
  \left[S_+^{[i]}S_0^{[i+1]} ,\mathcal{H}_{i,i+1}\right] |m_i,m_{i+1}\rangle 
  &=-(m_{i+1}+s)|m_i+1\rangle\otimes|m_{i+1}\rangle\\
 &\qquad+\sum_{k=1}^{m_i+1}\chi(m_i,m_{i+1},k)\,
 |m_i-k+1\rangle\otimes|m_{i+1}+k\rangle
 \\&\qquad+\sum_{k=1}^{m_{i+1}}
 \chi(m_{i+1},m_{i},k+1)\,
 |m_i+k+1\rangle\otimes|m_{i+1}-k\rangle\,.
 \end{split}
 \end{equation} 
 Finally,  isolating the $k=1$ term in the first sum and subsequently shifting $k$ by one yields
 \begin{equation}\label{eq:termm1}
  \begin{split}
    \left[S_+^{[i]}S_0^{[i+1]} ,\mathcal{H}_{i,i+1}\right] |m_i,m_{i+1}\rangle 
   &=-(m_{i+1}+s)|m_i+1\rangle\otimes|m_{i+1}\rangle\\
 &\qquad+\sum_{k=1}^{m_i}\chi(m_i,m_{i+1},k+1)
 |m_i-k\rangle\otimes|m_{i+1}+k+1\rangle\\
 &\qquad+\chi(m_i,m_{i+1},1)
 |m_i\rangle\otimes|m_{i+1}+1\rangle
 \\&\qquad+\sum_{k=1}^{m_{i+1}}
  \chi(m_{i+1},m_{i},k+1)
 |m_i+k+1\rangle\otimes|m_{i+1}-k\rangle \,.
  \end{split}
 \end{equation}

\medskip
\paragraph{The commutator $[S_+^{[k+1]}S_0^{[k]},\mathcal{H}_{k,k+1}]$.}
For the other commutator we have
  \begin{equation}
 \begin{split}
  &\left[S_+^{[i+1]}S_0^{[i]} ,\mathcal{H}_{i,i+1}\right] |m_i,m_{i+1}\rangle \\
  &=(m_{i}+s)(m_{i+1}+2s)\left(h_s(m_i)+h_s(m_{i+1})\right)|m_i\rangle\otimes|m_{i+1}+1\rangle\\
 &\qquad-\sum_{k=1}^{m_i}(m_{i}-k+s)(m_{i+1}+k+2s)
 \implicit{k}{m_i}
 |m_i-k\rangle\otimes|m_{i+1}+k+1\rangle
 \\&\qquad-\sum_{k=1}^{m_{i+1}}(m_{i}+k+s)(m_{i+1}-k+2s)
  \implicit{k}{m_{i+1}}
 |m_i+k\rangle\otimes|m_{i+1}-k+1\rangle\\
 &\qquad -(m_{i+1}+2s)(m_{i}+s)\mathcal{H}_{i,i+1}|m_i,m_{i+1}+1\rangle
 \end{split}
 \end{equation} 
 where
 \begin{equation}
\begin{split}
 \mathcal{H}_{i,i+1}|m_i\rangle\otimes|m_{i+1}+1\rangle&=\left(h_s(m_i)+h_s(m_{i+1}+1)\right)|m_i\rangle\otimes|m_{i+1}+1\rangle\\
 &\qquad-\sum_{k=1}^{m_i}
 \implicit{k}{m_i}
 |m_i-k\rangle\otimes|m_{i+1}+k+1\rangle
 \\&\qquad-\sum_{k=1}^{m_{i+1}}
  \implicit{k}{m_{i+1}+1}
 |m_i+k\rangle\otimes|m_{i+1}-k+1\rangle\\
 &\qquad-
 \implicit{m_{i+1}+1}{m_{i+1}+1}
 |m_{i+1}+m_i+1\rangle\otimes|0\rangle\,.
 \end{split}
\end{equation}

Then, proceeding as above, we end up with 
   \begin{equation}\label{eq:termm2}
 \begin{split}
  \left[S_+^{[i+1]}S_0^{[i]} ,\mathcal{H}_{i,i+1}\right] |m_i,m_{i+1}\rangle 
 &=-(m_{i}+s)|m_i\rangle\otimes|m_{i+1}+1\rangle\\
 &\qquad+\sum_{k=1}^{m_{i}}
  \chi(m_{i},m_{i+1},k+1)
 |m_i-k\rangle\otimes|m_{i+1}+k+1\rangle
 \\
 &\qquad+\sum_{k=1}^{m_{i+1}}
 \chi(m_{i+1},m_{i},k+1)
 |m_i+k+1\rangle\otimes|m_{i+1}-k\rangle
 \\
 &\qquad+
 \chi(m_{i+1},m_{i},1)
 |m_i+1\rangle\otimes|m_{i+1}\rangle\,.
 \end{split}
 \end{equation}

\medskip
\paragraph{Proof of Eq.\eqref{eq:together}.}
Taking the two terms \eqref{eq:termm1} and \eqref{eq:termm2} together the sums cancel. Further using $\chi(m_i,m_{i+1},1)=s+m_{i+1}-m_i$, we are left with
  \begin{equation}
 \begin{split}
 &\left( \left[S_+^{[i]}S_0^{[i+1]} ,\mathcal{H}_{i,i+1}\right]-  \left[S_+^{[i+1]}S_0^{[i]} ,\mathcal{H}_{i,i+1}\right] \right) |m_i,m_{i+1}\rangle \\
 &=-(m_{i}+2s)|m_i+1\rangle\otimes|m_{i+1}\rangle+(m_{i+1}+2s)
 |m_i\rangle\otimes|m_{i+1}+1\rangle\\
 &=(S_+^{[i+1]}-S_+^{[i]}) |m_i\rangle\otimes|m_{i+1}\rangle\,.
 \end{split}
 \end{equation} 
This completes the proof of Proposition~\ref{prop:symm}.

\end{proof}

\subsection{Derivation of the non-local transformation $\mathcal{W}$} 
\label{app:nonlocal}
Thanks to symmetry $Q''$ derived in the previous section we can now exhibit the transformation $\mathcal{W}$ that diagonalizes both boundaries of the Hamiltonian. 
We recall the definition (cf. \eqref{eq:H0-def})  
\begin{equation}
{H}^\circ=\mathcal{B}_1^\circ+\sum_{i=1}^{N-1}\mathcal{H}_{i,i+1}+\mathcal{B}_N^\circ
\end{equation} 
with 
\begin{eqnarray}
  \mathcal{B}_i^{\circ} &=&
  \psi(S_0^{[i]}+s)-\psi(2s), \qquad\qquad i\in\{1,N\}.
\nonumber\\
 \end{eqnarray}
 Then we have the following result.
\begin{proposition}[Non-local transformation $\mathcal{W}$]
\label{prop:non-local}
The matrix $\mathcal{W}$ given by
\begin{equation}
\label{eq:Wtrans}
\mathcal{W}=\sum_{k=0}^\infty (\rho_{L}-\rho_R)^k
\frac{Q_+^{k}}{k!}\frac{\Gamma(2 (S_0^{\mathrm{tot}}+s))}{\Gamma(k+2 (S_0^{\mathrm{tot}}+s))}\,,
\end{equation} 
where $Q_+$ is defined in \eqref{eq:Qplus} is such that
\begin{equation}
\label{eq:diagonalW}
{H}^\circ=\mathcal{W}^{-1}\,{H}''\,\mathcal{W}\;.
\end{equation} 
\end{proposition}

\begin{proof}
In view of Proposition \ref{prop:symm}, the operators
$H''$ and $Q''$ are simultaneously diagonalizable. Thus the  similarity transformation $\mathcal{W}$ in \eqref{eq:diagonalW}
that diagonalizes $H''$,  will also give
\begin{equation}\label{eq:WQ}
 Q^\circ = \mathcal{W}^{-1}Q''\mathcal{W}\;.
\end{equation} 
We make  the ansatz
\begin{equation}
 \mathcal{W}=1+\sum_{k=1}^\infty (\rho_{L}- \rho_{R})^kg_k
\end{equation} 
for the transformation $\mathcal{W}$. Substituting this ansatz into \eqref{eq:WQ} and using $Q'' = Q^{\circ} - (\rho_L-\rho_R) Q_+$, we obtain the commutation relation for the coefficients of the expansion
\begin{equation}
 [Q^\circ,g_k]=Q_+ g_{k-1}\,.
\end{equation} 
It follows that $g_k$ creates $k$ particles and the commutation relations can be brought to the form
\begin{equation}
g_k k (k+2 s+2S_0^{\mathrm{tot}}-1)=Q_+ g_{k-1}\,.
\end{equation} 
This difference equation can be solved and we find  \eqref{eq:Wtrans}.

\end{proof}

\subsection{Mapping non-equilibrium onto equilibrium}
\label{sec:equil}

With the transformation $\mathcal{W}$ in our hands we can now exhibit
the  transformation relating the non-equilibrium Hamiltonian
$H$ with boundary densities $\rho_L$ and $\rho_R$ to the equilibrium Hamiltonian $H^{\mathrm{eq}}$ with density $\rho$ at both sides:
\begin{equation}
{H}^{\mathrm{eq}}=\mathcal{B}_1^{\mathrm{eq}}+\sum_{i=1}^{N-1}\mathcal{H}_{i,i+1}+\mathcal{B}_N^{\mathrm{eq}}
\end{equation} 
with 
\begin{equation}
  \mathcal{B}_{i}^{\mathrm{eq}}=e^{-S_-^{[i]}}
  e^{\rho S_+^{[i]}}\Big(\psi(S_0^{[i]}+s)-\psi(2s)\Big)e^{-\rho S_+^{[i]}}e^{ S_-^{[i]}} \qquad\qquad i\in\{1,N\}\,.
\end{equation} 
\begin{proposition}[Mapping non-equilibrium onto equilibrium] We have
\begin{equation}
 {H}^{\mathrm{eq}}=P^{-1} \,{H} \,P
\end{equation} 
with
\begin{equation}
 P=e^{-S_-^{\mathrm{tot}}}e^{\rho_R S_+^{\mathrm{tot}}}\mathcal{W} e^{-\rho S_+^{\mathrm{tot}}}e^{S_-^{\mathrm{tot}}}.
\end{equation} 
\end{proposition}
\begin{proof}
Combining together Proposition \ref{prop:local}
and Proposition \ref{prop:non-local}
we have
\be
\label{eq:comb1}
{H}^{\circ}=\mathcal{W}^{-1}e^{-\rho_R S_+^{\mathrm{tot}}}e^{S_-^{\mathrm{tot}}}{H} e^{- S_-^{\mathrm{tot}}}e^{\rho_R S_+^{\mathrm{tot}}} \mathcal{W}
\ee 
Similarly, we can  relate the Hamiltonian with diagonal boundaries ${H}^{\circ}$ to the process ${H}^{\mathrm{eq}}$ by
\begin{equation}
\label{eq:comb2}
{H}^{\circ}=e^{-\rho S_+^{\mathrm{tot}}}e^{S_-^{\mathrm{tot}}}{H}^{\mathrm{eq}} e^{- S_-^{\mathrm{tot}}}e^{\rho S_+^{\mathrm{tot}}} \,.
\end{equation} 
The proof of the proposition follows from \eqref{eq:comb1} and \eqref{eq:comb2}.
\end{proof}

\section{Proof of the results}
\label{sec:proof-main-th}

In this section we prove our main results.
The non-equilibrium steady state is obtained by a sequence of transformations from the reference state 
\begin{equation}\label{eq:reference}
 |\Omega\rangle =|0\rangle\otimes|0\rangle\otimes\ldots\otimes|0\rangle
\end{equation} 
which is obviously a ground state of the Hamiltonian $H^{\circ}$, i.e. $H^\circ|\Omega\rangle=0$.  The non-equilibrium steady state is related to the reference state via the transformations introduced 
in Section~\ref{sec:transform}. In particular, denoting by $|\mu\rangle$ the column vector whose components furnish
the steady state probabilities, we have
\begin{equation}
 |\mu\rangle =e^{-S_-^{\mathrm{tot}}}e^{\rho_R S_+^{\mathrm{tot}}}{\mathcal{W}}|\Omega\rangle\;. 
\end{equation} 

\paragraph{Section organization.}
We prove Theorem~\ref{theo:factorial} in Section~\ref{sec:proof-theo-factorial} via a series of Lemmata. We first compute the action of $\mathcal{W}$ on the reference state in Lemma~\ref{lem:actW}. This yields the state $|\mu''\rangle=\mathcal{W}|\Omega\rangle$ that is an eigenvector of ${H}''$ given in \eqref{eq:Hpp} with vanishing eigenvalue. The result can then be used to obtain the eigenstate $|\mu'\rangle=\exp[{\rho_R S_+^{\mathrm{tot}}}]|\mu''\rangle$ for ${H}'$ in \eqref{eq:Hp} with zero eigenvalue, see Lemma~\ref{lem:gsHprime}. This eigenvector is directly connected to the absorption probabilities of the dual process
and the factorial moments, as explained in Lemma ~\ref{lem:factorial-mom-proof},
thus allowing to complete the proof of Theorem~\ref{theo:factorial}.  

We then prove the corollaries of Theorem~\ref{theo:factorial} in the remaining paragraphs. 
In Section \ref{sec:proof-coordinate} we prove Corollary~\ref{cor:fac-mom}
yielding the factorial moments in coordinate form. This is obtained 
introducing Young diagrams and using a symmetry property that allows
to exchange rows and columns.
Finally Corollary \ref{theo:steady} and Corollary \ref{theo:loc}
are proved, respectively, in Section~\ref{sec:closed} and in Section~\ref{sec:localeq}.

\subsection{Proof of Theorem \ref{theo:factorial}: factorial moments.}
\label{sec:proof-theo-factorial}

\begin{lemma}[Ground state $\mu''$ of $H''$]
\label{lem:actW}
For all $m\in{\mathscr C}_N$ we have
\be
\label{eq:mupp}
\mu''(m) = 
(\rho_{L}-\rho_{R})^{|m|}\frac{\Gamma(2s (N+1))}{\Gamma(|m|+2s(N+1))} 
 \prod_{i=1}^N\frac{\Gamma(2s+m_i)}{\Gamma(2s)\Gamma(1+m_i)}\frac{\Gamma\left(2s(N+1-i)+\sum_{k=i}^{N}m_k\right)}{\Gamma\left(2s(N+1-i)+\sum_{k=i+1}^{N}m_k\right)}
\ee
where we recall the notation $|m|= \sum_{i=1}^N m_i$.
\end{lemma}
\begin{proof}
We have $\mu''(m) = \langle m |\mu''\rangle= \langle  m |\mathcal{W}|\Omega\rangle$, so that we need to compute  
the action of the non-local transformation $\mathcal{W}$ 
on the reference state \eqref{eq:reference}.
We employ the explicit expression of $\mathcal{W}$ obtained in Proposition \ref{prop:non-local} as a power series in the operator $Q_+$.
Using the fact that the reference state $|\Omega\rangle$  contains no particles, and that  the action of $Q_+$ amounts to the creation of a single particle (see \eqref{eq:Qplus}), { we get that the components of the state $|\mu''\rangle$ 
are proportional to the matrix elements  $\langle m| Q_+^{|m|}|\Omega\rangle$. More precisely we have}
\begin{equation}
\begin{split}
\mu''(m)&=\sum_{k=0}^\infty 
\frac{\Delta ^k}{k!}\frac{\Gamma(2s (N+1))}{\Gamma(k+2 s(N+1))}\langle  m|Q_+^{k}|\Omega\rangle=
\frac{\Delta ^{|m|}}{|m|!}\frac{\Gamma(2s (N+1))}{\Gamma(|m|+2 s(N+1))}\langle  m|Q_+^{|m|}|\Omega\rangle\,.
\end{split}
\end{equation} 
We shall prove below that such matrix elements are given by
\begin{equation}\label{eq:qpm}
\langle  m|Q_+^{|m|}|\Omega\rangle =|m|!
 \prod_{i=1}^N\frac{\Gamma(2s+m_i)}{\Gamma(2s)\Gamma(1+m_i)}\frac{\Gamma\left(2s(N-i+1)+\sum_{k=i}^{N}m_k\right)}{\Gamma\left(2s(N-i+1)+\sum_{k=i+1}^{N}m_k\right)}\,.
\end{equation} 
First we remark that the matrix elements can be  defined recursively (with
$\langle m|Q_+^{|m|}|\Omega\rangle = 1$ when $|m|=0$) via
\begin{equation}\label{eq:recurs}
\begin{split}
 \langle m|Q_+^{|m|}|\Omega\rangle =\sum_{p=1}^N\langle  m|Q_+| m-\delta_p\rangle \langle  m-\delta_p|Q_+^{|m|-1}|\Omega\rangle\,,
 \end{split}
\end{equation}
where $| m-\delta_p\rangle$ denotes the state that is obtained
from $| m\rangle$ by removing a particle at site $p$.
Here we employed a resolution of the identity 
and we used again the fact that the action of $Q_+$ is the creation of a single particle.
Thus proving \eqref{eq:qpm} is equivalent to showing that \eqref{eq:qpm} solves the recursive equation \eqref{eq:recurs}. To this aim 
we first compute the left factor in the right hand side of \eqref{eq:recurs}. We get
\begin{equation}
 \langle  m|Q_+| m-\delta_p\rangle =\left(2s+m_p-1\right)\left(m_p-1+2s(N-p+1)+2\sum_{j=p+1}^Nm_j\right)\,,
\end{equation} 
which follows from the explicit expression of the operator $Q_+$ in \eqref{eq:Qplus}. 
Next, noting that 
\begin{equation}
\begin{split}
 \frac{\langle  m-\delta_p|Q_+^{|m|-1}|\Omega\rangle }{\langle  m|Q_+^{|m|}|\Omega\rangle}&=\frac{1}{|m|}\frac{m_p}{2s+m_p-1}\frac{\prod_{i=1}^{p-1}\left(2s(N-i+1)-1+\sum_{j=i+1}^N m_j\right)}{\prod_{i=1}^p\left(2s(N-i+1)-1+\sum_{j=i}^N m_j\right)}\\
 &=\frac{1}{|m|(2sN-1+|m|)}\frac{m_p}{2s+m_p-1}\prod_{i=1}^{p-1}\frac{\left(2s(N-i+1)-1+\sum_{j=i+1}^N m_j\right)}{ \left(2s(N-i)-1+\sum_{j=i+1}^N m_j\right)}\,,
 \end{split}
\end{equation} 
which is a direct consequence of \eqref{eq:qpm},   the recursive
relation \eqref{eq:recurs} reduces to
\begin{equation}\label{eq:toshow}
\begin{split}
 |m|(2sN-1+|m|)=\sum_{p=1}^N m_p&\left(m_p-1+2s(N-p+1)+2\sum_{j=p+1}^Nm_j\right)\\ &\qquad\times\prod_{i=1}^{p-1}\frac{1-2s(N-i+1)-\sum_{j=i+1}^N m_j}{1-2s(N-i)-\sum_{j=i+1}^N m_j}\,.
 \end{split}
\end{equation} 
Thus, to show \eqref{eq:qpm}, we need to verify this relation for all $m\in{\mathscr C}_N$.
To prove this formula it is convenient to introduce the variables 
\begin{equation}\label{eq:minw}
 w_i=1-2s(N-i+1)-\sum_{j=i}^N
m_j \qquad\qquad \text{for}\quad i=1,2,\ldots,N\,,
\end{equation} 
with inverse map  given by
\begin{align}
 m_i&=w_{i+1}-w_i-2s \qquad\qquad \text{for}\quad i=1,2,\ldots,N-1\,,\\
 m_N&=1-w_N-2s\,.
\end{align} 
Under this change of variables \eqref{eq:toshow} above becomes
\begin{equation}
\begin{split}
(  w_1+2sN-1) w_1&=(  w_{1}-  w_{2}+2s)(-1+  w_1+  w_{2}+2s(N-1))\\
 &\quad\,+\sum_{p=2}^{N-1} (  w_{p}-  w_{p+1}+2s)(-1+  w_p+  w_{p+1}+2s(N-p))\prod_{i=1}^{p-1}\frac{  w_{i+1}-2s}{  w_{i+1}}\\
 &\quad\, +  w_N(  w_N+2s-1)\prod_{i=1}^{N-1}\frac{  w_{i+1}-2s}{  w_{i+1}}\,.
 \end{split}
\end{equation} 
Subtracting $(w_1+2sN-1) w_1$ on both sides the $w_1$ dependence drops and we remain with
\begin{equation}
\label{eq:temp}
 \begin{split}
0&=(2s-  w_{2})(-1+  w_{2}+2s(N-1))\\
 &\quad\,+\sum_{p=2}^{N-1}(  w_p+2 s) (2s(N-p) +w_p-1)\prod_{i=1}^{p-1}\frac{  w_{i+1}-2s}{  w_{i+1}}\\
 &\quad\,-\sum_{p=2}^{N-1}  w_{p+1}(  w_{p+1}-1+2 s (N-p-1))\prod_{i=1}^{p-1}\frac{  w_{i+1}-2s}{  w_{i+1}}\\
 &\quad\,+(  w_N(  w_N -1)-2s(2s-1))\prod_{i=1}^{N-2}\frac{  w_{i+1}-2s}{  w_{i+1}}\,.
 \end{split}
\end{equation} 
Here we also separated the $w_p$ and $w_{p+1}$ dependencies in the sum over $p$.
We remark now that in \eqref{eq:temp} the terms containing $w_2$, as well as those containing $w_N$, do cancel. After shifting by one the first sum over $p$
we finally get
\begin{equation}\label{eq:eq1}
\begin{split}
 2 s (2s( N -2)-1)=
 &\sum_{p=2}^{N-2}\frac{4 s^2 (w_{p+1}+2s( N- p -1)-1)}{w_{p+1}}\prod_{i=2}^{p-1}\frac{  w_{i+1}-2s}{  w_{i+1}}+2s(2s-1)\prod_{i=2}^{N-2}\frac{  w_{i+1}-2s}{  w_{i+1}}\,.
 \end{split}
\end{equation} 
This equation can be shown to be true introducing the function
\begin{equation}
 f_k= 2 s (2 s (N-k+1)-1)\,,\qquad k\in\N_0\,,
\end{equation} 
and noticing that it satisfies the recursion equation 
\begin{equation}\label{eq:iter}
f_k=\frac{4 s^2 (w_{k}+2 s(N-k)-1)}{w_{k}}+f_{k+1} \frac{w_{k}-2 s}{w_{k}}\,.
\end{equation} 
In fact this relation is true for $w_k\in \C$ and  not only for the specific sequence  $w_k$ defined in \eqref{eq:minw}.
Iterating \eqref{eq:iter} from $k=3$ up to $k=N$ we obtain 
\begin{equation}\label{eq:id18}
\begin{split}
f_3
 &=\sum_{k=3}^{N-1}\frac{4 s^2 (w_{k}+2 s(N-k)-1)}{w_{k}}\prod_{j=3}^{k-1}\frac{w_{j}-2s}{w_{j}}+f_{N}\prod_{j=3}^{N-1}\frac{w_{j}-2 s}{w_{j}}\,.
 \end{split}
\end{equation} 
After shifting the index of the $w$'s the equation  \eqref{eq:id18} reduces to \eqref{eq:eq1}. This concludes the proof of \eqref{eq:toshow} and thus the proof of 
Lemma \ref{lem:actW} is obtained.
\end{proof}
It is easy to compute now the ground state of $H'$.
\begin{lemma}[Ground state $\mu'$ of $H'$]
\label{lem:gsHprime}
For all $m\in{\mathscr C}_N$ we have
\begin{eqnarray}
\label{eq:muprime}
\mu'(m) 
&=& 
\sum_{\eta\in\mathscr{C}_N}\rho_R^{|m|-|\eta|}(\rho_L-\rho_R)^{|\eta|}\frac{\Gamma(2s (N+1))}{\Gamma(|\eta|+2s(N+1))} \times\\
& &\qquad\qquad
\prod_{i=1}^N\frac{1}{\eta_i! (m_i-\eta_i)!} \frac{\Gamma(m_i+2s)}{\Gamma(2s)}\frac{\Gamma\left(2s(N+1-i)+\sum_{k=i}^{N}\eta_k\right)}{\Gamma\left(2s(N+1-i)+\sum_{k=i+1}^{N}\eta_k\right)}\,.
\nonumber\end{eqnarray} 

\end{lemma}
\begin{proof}
Recalling that $|\mu'\rangle=\exp[{\rho_R S_+^{\mathrm{tot}}}]|\mu''\rangle$ we have
\be
\mu'(m) = \sum_{\eta\in {\mathscr C}_N}  \langle m |\exp[{\rho_R S_+^{\mathrm{tot}}}]|\eta\rangle \, \mu''(\eta)\,.
\ee
Inserting the expression \eqref{eq:mupp} for $\mu''$ and using the relation
\begin{equation}
\langle  m|e^{\rho_R S_+^{\mathrm{tot}}} | \eta\rangle = \rho_R^{|m|-|\eta|}\prod_{i=1}^N \frac{1}{(m_i-\eta_i)!} \frac{\Gamma(m_i+2s)}{\Gamma(\eta_i+2s)}
\end{equation}  
we immediately find \eqref{eq:muprime}.

\end{proof}

\begin{lemma}[Factorial moments from the ground state $\mu'$]
\label{lem:factorial-mom-proof}
For a multi-index $\xi = (\xi_1,\ldots,\xi_N)\in{\mathscr C}_N$, the scaled factorial moments of order $|\xi| = \sum_{i=1}^N \xi_i$ of the non-equilibrium steady state are given by
\begin{eqnarray}
\label{eq:facmom-mup}
G(\xi)
&=& 
\sum_{n=0}^{|\xi|}\rho_R^{|\xi|-n}(\rho_L-\rho_R)^{n}\frac{\Gamma(2s (N+1))}{\Gamma(n+2s(N+1))} \times\\
& &\qquad\qquad
\sum_{\underset{\eta_1+\ldots+\eta_{\scalebox{0.5} N}=n}{\eta\in\mathscr{C}_N}} \prod_{i=1}^N{\xi_i \choose \eta_i}\frac{\Gamma\left(2s(N+1-i)+\sum_{k=i}^{N}\eta_k\right)}{\Gamma\left(2s(N+1-i)+\sum_{k=i+1}^{N}\eta_k\right)}
\nonumber\end{eqnarray} 
\end{lemma}
\begin{proof}
We start by ``opening'' the expression of $\mu'$ in \eqref{eq:muprime}.
We slice the sum $\sum_{\eta\in\mathscr{C}_N}$ over all configurations $\eta=(\eta_1,\ldots,\eta_N)$ with a fixed number of particles
$|\eta|={\eta_1+\ldots+\eta_N}=n$, where $n$ can range from $0$ to $|m|$ as a consequence of the factors $(m_i-\eta_i)!$ in the denominator of \eqref{eq:muprime}.
We thus get
\begin{eqnarray}
\label{eq:mup}
\mu'(m) 
&=& 
\sum_{n=0}^{|m|}\rho_R^{|m|-n}(\rho_L-\rho_R)^{n}\frac{\Gamma(2s (N+1))}{\Gamma(n+2s(N+1))} \times\\
& &\qquad\qquad
\sum_{\underset{\eta_1+\ldots+\eta_{\scalebox{0.5} N}=n}{\eta\in\mathscr{C}_N}} \prod_{i=1}^N\frac{1}{\eta_i! (m_i-\eta_i)!} \frac{\Gamma(m_i+2s)}{\Gamma(2s)}\frac{\Gamma\left(2s(N+1-i)+\sum_{k=i}^{N}\eta_k\right)}{\Gamma\left(2s(N+1-i)+\sum_{k=i+1}^{N}\eta_k\right)}
\nonumber\end{eqnarray} 
Next we observe that
\be
\label{eq:gdmu}
G(\xi) = d^{\rm tot} (\xi,\xi) \mu'(\xi)
\ee
where 
\begin{equation}
 d^{\,\rm tot}=\prod_{i=1}^{N}d^{[i]}\,,
\end{equation} 
with $d^{[i]}$ given in \eqref{eq:r} and satisfying \eqref{eq:trans}.
This follows from the  algebraic rewriting of the definition \eqref{eq:fac-moments} of the scaled factorial moments, i.e
\begin{equation}
G(\xi) =\langle \mu | \prod_{i=1}^N d^{[i]}e^{S_+^{[i]}} | \xi\rangle
\end{equation}
and the relation $|\mu \rangle = e^{-S_-^{\mathrm tot}} |\mu'\rangle$
which is implied by \eqref{eq:h prime transform}.
The proof of \eqref{eq:facmom-mup} is obtained by combining \eqref{eq:mup} and \eqref{eq:gdmu}.

\end{proof}

\begin{proof}[{\bf Proof of Theorem \ref{theo:factorial}}]
The proof of \eqref{ed2} and \eqref{eq:gfunc} follows now by simplifying the Gamma functions in \eqref{eq:facmom-mup}. 
Indeed, for integers $\ell_1,\ell_2\in\N$ a ratio of Gamma functions 
can be written as a  product
\begin{equation}
\label{eq:gammmaratio1}
\frac{\Gamma(\ell_1+\ell_2)}{\Gamma(\ell_1)} = \prod_{j=1}^{\ell_2} (\ell_1-j+\ell_2)\;.   
\end{equation}
Therefore the fractions of Gamma functions inside the sum in \eqref{eq:facmom-mup} can be written as 
\begin{equation}\label{eq:finprod} \frac{\Gamma\left(2s(N-i+1)+\sum_{k=i}^{N}\eta_k\right)}{\Gamma\left(2s(N-i+1)+\sum_{k=i+1}^{N}\eta_k\right)}= \prod_{j=1}^{\eta_i}\left(2s(N-i+1)-j+\sum_{k=i}^{N}\eta_k\right)\,.
\end{equation}
Further noting that $n=|\eta|$, the fraction of Gamma functions outside the sum in \eqref{eq:facmom-mup} can be written as a telescopic product so that 
\begin{equation}
 \frac{\Gamma(2s(N+1)}{\Gamma(2s(N+1)+|\eta|)}= \prod_{i=1}^N \frac{\Gamma(2s(N+1)+\sum_{k=i+1}^{N}\eta_k)}{\Gamma(2s(N+1)+ \sum_{k=i}^{N}\eta_k)}=\prod_{i=1}^N\prod_{j=1}^{\eta_i}\frac{1}{2s(N+1)-j+\sum_{k=i}^{N}\eta_k}\,.
\end{equation} 
Inserting the last two expressions into \eqref{eq:facmom-mup} we arrive at \eqref{ed2} and \eqref{eq:gfunc}.

If the spin takes an half-integer value  $s\in \N/2$ then the expression for $g_{\xi}$ in \eqref{eq:gfunc} 
can be simplified in the same spirit. We now write
\begin{eqnarray}
&& \frac{\Gamma(2s(N+1))}{\Gamma(2s(N+1)+|\eta|)}\prod_{i=1}^N  \frac{\Gamma\left(2s(N+1-i)+\sum_{k=i}^{N}\eta_k\right)}{\Gamma\left(2s(N+1-i)+\sum_{k=i+1}^{N}\eta_k\right)}
 \nonumber \\
&& \qquad\qquad = \frac
{\Gamma(2s(N+1))}
{\Gamma(2s)  } 
 \frac
{\prod_{i=1}^N\Gamma\left(2s(N+1-i)+\sum_{k=i}^{N}\eta_k\right)}
{ \prod_{i=1}^{N }\Gamma\left(2s(N+2-i)+\sum_{k=i}^{N}\eta_k\right) }\,.
\end{eqnarray}
Applying again \eqref{eq:gammmaratio1} this produces 
\begin{equation}
\label{eq: express-onexxx}
 g_{\xi}(n)=
\sum_{\underset{\eta_1+\ldots+\eta_{\scalebox{0.5} N}=n}{(\eta_1,\ldots,\eta_{\scalebox{0.5} N}) \in \N_0^{\scalebox{0.5} N}}}\prod_{i=1}^N \binom{\xi_i}{\eta_i} \prod_{j=1}^{2s}\frac{2s(N+2-i) - j}{2s(N+2-i) - j  +\sum_{k=i}^{N}\eta_k}\,,
\end{equation}
which justifies the expression 
\eqref{eq: express-one} for $s=1/2$.
\end{proof}

\subsection{Proof of Corollary \ref{cor:fac-mom}: factorial moments in coordinate form.}
\label{sec:proof-coordinate}

The key to the proof of Corollary~\ref{cor:fac-mom}
is the following lemma, which introduces Young diagrams to describe the change of coordinates between occupations and ordered positions.
\begin{lemma}[Functions on Young diagrams]
\label{lemma:fct-young}
Let $m\in\mathscr{C}_N$ be a configuration of particle numbers at each site
and let $x=(x_1,\ldots,x_{|m|})$ with $1\le x_1\le\ldots\le x_{|m|}\le N$
be the corresponding vector of particle positions. Then
\begin{equation}
\label{eq:occ-to-pos}
 \prod_{k=1}^N\prod_{i=1}^{m_k}\frac{2s(N+1-k)-i+\sum_{l=k}^{N}m_l}{2s(N+1)-i+\sum_{l=k}^{N}m_l}=\prod_{i=1}^{|m|}\frac{|m|-i+2s(N+1-x_{i})}{|m|-i+2s(N+1)}\,.
\end{equation} 
\end{lemma}
\begin{proof}
We start by noticing that we can rewrite the right hand side as a telescopicing product 
\begin{equation}
\prod_{i=1}^{|m|}\frac{|m|-i+2s(N+1-x_{i})}{|m|-i+2s(N+1)}=\prod_{i=1}^{|m|}\prod_{j=1}^{x_i}f_s(i,j)
\end{equation} 
with
\begin{equation}
 f_s(i,j)=\frac{2s(N+1-j)-i+|m|}{2s(N+2-j)-i+|m|}.
\end{equation} 
Similarly for the left hand side we write
\begin{equation}
\begin{split}
 \prod_{k=1}^N\prod_{i=1}^{m_k}\frac{2s(N+1-k)-i+\sum_{l=k}^{N}m_l}{2s(N+1)-i+\sum_{l=k}^{N}m_l}&=\prod_{k=1}^N\prod_{i=1}^{m_k}\prod_{j=1}^k\frac{2s(N+1-j)-i+\sum_{l=k}^{N}m_l}{2s(N+2-j)-i+\sum_{l=k}^{N}m_l}\\
 &=\prod_{k=1}^N\prod_{i=1}^{m_k}\prod_{j=1}^k\frac{2s(N+1-j)-i+|m|-\sum_{l=1}^{k-1}m_l}{2s(N+2-j)-i+|m|-\sum_{l=1}^{k-1}m_l}\\
 &=\prod_{k=1}^N\prod_{i=1}^{m_k}\prod_{j=1}^{k}f_s(i+\sum_{l=1}^{k-1}m_l,j)\;.
 \end{split}
\end{equation} 
Recall that  $x=(x_1,\ldots,x_{|m|})$ with 
$1\le x_1\le\ldots\le x_{|m|}\le N$
denotes the set of ordered positions of the particles. 
So we can write
\begin{equation}
x= (\underbrace{1,\ldots,1}_{m_{1}},\underbrace{2,\ldots,2}_{m_{2}},\ldots,\underbrace{N,\ldots,N}_{m_{N}})\,,
\end{equation} 
see Figure~\ref{fig:tab} where a Young diagram corresponding to certain positions  and occupations is depicted.
As a consequence  we find that for any $f_s(i,j)$ assigned to a box $(i,j)$ in the $i$th column and $j$th row of the Young diagram we get
\begin{equation}
\begin{split}
 \prod_{i=1}^{|m|}\prod_{j=1}^{x_i}f_s(i,j)&=\prod_{i=1}^{m_1}\prod_{j=1}^{1}f_s(i,j)\prod_{i=m_1+1}^{m_1+m_2}\prod_{j=1}^{2}f_s(i,j)\cdots\prod_{i=m_1+\ldots+m_{N-1}+1}^{m_1+\ldots+m_{N}}\prod_{j=1}^{N}f_s(i,j)\\&=\prod_{k=1}^N\prod_{i=1+\sum_{l=1}^{k-1} m_l}^{\sum_{l=1}^km_l}\prod_{j=1}^{k}f_s(i,j)\\
 &=\prod_{k=1}^N\prod_{i=1}^{m_k}\prod_{j=1}^{k}f_s(i+\sum_{l=1}^{k-1}m_l,j)\,.
 \end{split}
 \end{equation} 
This proves \eqref{eq:occ-to-pos}.
We remark that, up to a factor $(\rho_{L}-\rho_{R})^{|m|}$, \eqref{eq:occ-to-pos} is exactly the value of the ground state $\mu''$. 
\end{proof}
\begin{figure}
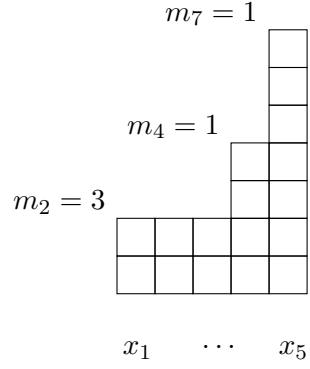

\begin{center}
  \parti
\end{center}
\caption{Example of a Young diagram corresponding to the positions $x=(2,2,2,4,7)$ and occupations $m=(0,3,0,1,0,0,1)$.}
\label{fig:tab}
\end{figure}

\begin{proof}[{\bf Proof of Corollary~\ref{cor:fac-mom}}]
The proof is obtained by going to the coordinate description of the dual process, i.e.
the multi-index $(\xi_1,\ldots,\xi_N)\in\N_0^N$ is replaced by the position of the dual particles
$x_1\le x_2 \le \ldots \le x_{|\xi|}$ .

\begin{figure}
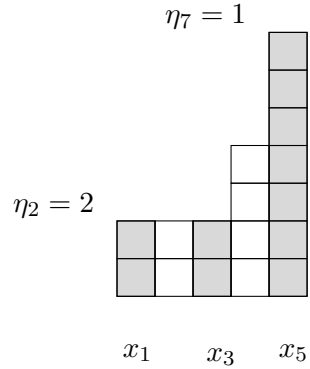

\begin{center}
  \subparti
\end{center}
\caption{Example of a Young subdiagram corresponding to the positions $\hat x=(2,2,7)$ and occupations $\eta=(0,2,0,0,0,0,1)$ }
\label{fig:tab2}
\end{figure}

Consider a Young subdiagram  $\hat x_{n}=(x_{\alpha_1},\ldots,x_{\alpha_n})$ with $n$ columns and  $1\le \alpha_1<\ldots<\alpha_n\le | \xi|$ of a Young diagram labelled by  $x$.   Let $\eta$ denote the occupation numbers of the subdiagram with $\eta_i\le \xi_i$. The number of particles corresponds to the number of colums in the diagram and we can write  $n=|\eta|$, see Figure~\ref{fig:tab2} for an example of a Young subdiagram. With these definitions we obtain as above
\begin{equation}
\prod_{i=1}^{|\eta|}\prod_{j=1}^{x_{\alpha_i}}f_s(i,j)=\prod_{i=1}^{\eta_k}\prod_{j=1}^k f_s(i+\sum_{l=1}^{k-1}\eta_l,j)\,.
\end{equation} 
Now we   sum over all Young subdiagrams with $n$ columns on both sides of the equation.    We remark that in position variables subdiagrams of same shape give same contribution, i.e. for different $i_\alpha$ but same $x_{i_\alpha}$. This is obvious for the occupation variables but leads to a combinatorial prefactor.  We find
\begin{equation}
 \sum_{1\leq \alpha_1< \ldots<\alpha_{n}\leq |  \xi|}\prod_{i=1}^{|\eta|}\prod_{j=1}^{x_{\alpha_i}}f(i,j)=\sum_{{\underset{\eta_1+\ldots+\eta_{\scalebox{0.5} N}=n}{(\eta_1,\ldots,\eta_{\scalebox{0.5} N}) \in \N_0^{\scalebox{0.5} N}}}}\prod_{k=1}^N\binom{\xi_k}{\eta_k}  \prod_{i=1}^{\eta_k}\prod_{j=1}^k f(i+\sum_{l=1}^{k-1}\eta_l,j)\,.
\end{equation} 
This concludes the proof.

\end{proof}

\subsection{Proof of Corollary \ref{theo:steady}:  stationary state}
\label{sec:closed}

\begin{proof}[{\bf Proof of Corollary \ref{theo:steady}}]
The expression of the stationary measure in \eqref{eq:stead1} - \eqref{eq:stead2} is obtained by applying the inversion formula \eqref{eq:measure} to Theorem \ref{theo:factorial}.
Equivalently, from the algebraic perspective, the stationary state is given by 
\be
\mu(m) = \sum_{\xi\in {\mathscr C}_N}  \langle m |\exp[{-S_-^{\mathrm{tot}}}]|\xi\rangle \, \mu'(\xi)\,.
\ee
Inserting the expression \eqref{eq:muprime} for $\mu'$ and using the relation
\begin{equation}
\langle  m|e^{- S_-^{\mathrm{tot}}} | \xi\rangle = (-1)^{|\xi|-|m|}\prod_{i=1}^N \frac{1}{(\xi_i-m_i)!} \frac{\xi_i!}{m_i!}
\end{equation}  
we  find \eqref{eq:stead1} and \eqref{eq:stead2}.
\end{proof}

\subsection{Proof of Corollary \ref{theo:loc}: local equilibrium}
\label{sec:localeq}

\begin{proof}[{\bf Proof of Corollary \ref{theo:loc}}]
To prove item i) it is enough to study the convergence of the scaled factorial moments.
For all configurations $(\xi_1,\ldots,\xi_N)\in\N^N$ made
of $|\xi|$ dual particles located at positions $1\le x_1 \le x_2 \le \ldots \le x_{|{ \xi}|}\le N$ 
and for $u\in(0,1)$, we define the translated configuration $(\xi_1^u,\ldots,\xi_N^u)\in\N^N$ with
$|\xi|$ dual particles located at positions $x_1 + [uN]\le x_2+ [uN] \le \ldots \le x_{|{ \xi}|}+ [uN]$.
Then we need to prove that for all configurations
\be
\lim_{N\to\infty} G(\xi_1^u,\ldots,\xi_N^u) = [\rho(u)]^{|\xi|}
\ee
with $\rho(u)$ given in \eqref{eq:linear}.
This follows from Corollary \ref{cor:fac-mom}, for the explicit formula for $g_{\xi}(n)$ in \eqref{eq: express-two} gives
\begin{equation}
\lim_{N\to\infty} g_{ \xi^u}(n)= {|\xi| \choose n} (1-u)^n
\end{equation} 
and therefore \eqref{ed} yields
\begin{eqnarray}
\lim_{N\to\infty} G(\xi_1^u,\ldots,\xi_N^u) 
& = &
\sum_{n=0}^{|\xi|}\rho_R^{|\boldsymbol  \xi|-n}(\rho_L-\rho_R)^n {|\xi| \choose n} (1-u)^n \nonumber \\
& = & 
 [\rho_R + (\rho_L-\rho_R)(1-u)]^{|\xi|} \\
& = & 
[\rho(u)]^{|\xi|}
\end{eqnarray}
The proof of item ii) follows from the definition
\eqref{eq:currbond} of the current. Indeed, Theorem \ref{theo:factorial},
specialized to one dual particle yields that the average current is the same for all bonds and is given by
 \be
 J_{i,i+1} = - 2s \frac{\rho_R-\rho_L}{L+1}\,, \qquad\qquad \forall i,i+1\,.
 \ee
Fick's law \eqref{eq:fick} then obviously follows.
\end{proof}

\paragraph{Acknowledgements} 
We thank Jorge Kurchan for several enlightening discussions.  RF  thanks  Istv\'an M. Sz\'ecs\'enyi for very useful comments.
RF is supported by the German research foundation (Deutsche Forschungsgemeinschaft DFG) Research Fellowships Programme 416527151.

{
\small
\bibliographystyle{utphys2}
\bibliography{refs}
}

\end{document}